\documentclass[pdftex,11pt]{article}
\usepackage{enotez}
\usepackage{amsmath}
\usepackage{amsfonts}
\usepackage{amssymb}
\usepackage{amsthm}
\usepackage{physics}
\usepackage{color}
\usepackage{cite}
\usepackage{fullpage}
\usepackage[colorlinks=true,linkcolor=blue,citecolor=magenta,linktocpage=true]{hyperref}
\usepackage{titlesec}
\titleformat*{\section}{\normalsize\bfseries}
\titleformat*{\subsection}{\normalsize\bfseries}
\titleformat*{\subsubsection}{\normalsize\bfseries}

\makeatletter
\renewcommand{\@dotsep}{1000}

\usepackage[normalem]{ulem}

\everymath{\displaystyle}

\def\be#1\ee{\begin{align}#1\end{align}}
\def\nn{\nonumber}
\def\q{\qquad}
\def\f{\frac}

\usepackage{bbm} % For the identity matrix
\newcommand{\one}{\mathbbm{1}}

\newcommand{\pr}[1]{\left(#1\right)}
\newcommand{\vr}[1]{\left|#1\right|}
\newcommand{\vvr}[1]{\left\|#1\right\|}

\newcommand{\cur}[1]{\left\{#1\right\}}
\newcommand{\ang}[1]{\left\langle#1\right\rangle}
\newcommand{\half}{\frac{1}{2}}
\newcommand{\ihalf}{\frac{i}{2}}
\newcommand{\pd}{\partial}

\newcommand{\R}{\mathbb{R}}
\newcommand{\Z}{\mathbb{Z}}
\newcommand{\C}{\mathbb{C}}
\newcommand{\X}{\mathbb{X}}
\newcommand{\Y}{\mathbb{Y}}
\newcommand{\K}{\mathbb{K}}

\newcommand{\WL}{{\mathcal{W}_\Lambda}}
\newcommand{\PS}{\mathcal{P}}
\newcommand{\LS}{\mathcal{L}}
\newcommand{\HS}{\mathcal{H}}
\newcommand{\WS}{\mathcal{W}}

\newcommand{\Weyl}[1]{\hat{W}_{#1}}

\newcommand{\tell}{{\tilde{\ell}}}
\newcommand{\tambda}{{\tilde{\lambda}}}

\newtheorem{propo}{Proposition}

\numberwithin{equation}{section}

\newcommand{\putinparen}[1]{\textsuperscript{\!\!\{#1\}}}
\DeclareInstance{enotez-list}{custom}{paragraph}{format=\normalsize,number=\{#1\}\,\,}
\setenotez{list-name=,mark-cs=\putinparen,mark-format=,backref=true,split=section}

\begin{document}

\title{\Large{\textbf{\sffamily Modular polymer representations of the Weyl algebra}}}
\author{\sffamily Yigit Yargic$^{1,2}$\thanks{yyargic@perimeterinstitute.ca} \, \& Marc Geiller$^3$\thanks{marc.geiller@ens-lyon.fr}}
\date{\small{\textit{
$^1$Perimeter Institute for Theoretical Physics,\\ 31 Caroline Street North, Waterloo, Ontario, N2L 2Y5, Canada\\[5pt]
$^2$Department of Physics and Astronomy, University of Waterloo,\\ 200 University Avenue West, Waterloo, Ontario, N2L 3G1, Canada\\[5pt]
$^3$Univ Lyon, ENS de Lyon, Univ Claude Bernard Lyon 1,\\ CNRS, Laboratoire de Physique, UMR 5672, F-69342 Lyon, France\\}}}

\maketitle

\begin{abstract}

One of the key conceptual challenges in quantum gravity is to understand how quantum theory should modify the very notion of spacetime. One way to investigate this question is to study the alternatives to Schr\"odinger quantum mechanics. The polymer representation, inspired by loop quantum gravity, can be understood as capturing features of discrete spatial geometry. The modular representation, on the other hand, has a built-in unification of position and momentum polarizations via a length scale. In this paper, we introduce the modular polymer representations of the Weyl algebra, in which neither position nor momentum exists as a well-defined operator. As inequivalent representations, they are candidates for describing new physics. We illustrate this by studying the dynamics of the harmonic oscillator as an example, with the prospect of eventually applying this representation to quantum cosmology.

\end{abstract}

\thispagestyle{empty}
\newpage
\setcounter{page}{1}

\hrule
\tableofcontents
%\addtocontents{toc}{\protect\setcounter{tocdepth}{2}} 
\vspace{0.7cm}
\hrule

%%%%%%%%%%%%%%%%%%%%%%%%%%%%%%%%%%%%%%%%%%%%%%

\section{Motivations}

It is reasonable to expect that properly understanding and formulating a theory of quantum gravity will require a rethinking of the concept of spacetime itself. Taking at face value the lesson from Einstein's theory of general relativity, which is that gravity is encoded in the geometry of spacetime, one is naturally led to the idea that a quantum theory of gravity should encode some notion of quantum spacetime geometry, and this in a background-independent manner. This viewpoint is at the core of loop quantum gravity (LQG hereafter) and related discrete approaches, which indeed provide a notion of quantum geometry \cite{Thiemann:2001yy,Ashtekar:2004eh,Oriti:2018dsg}.

Instead of focusing on the quantization of the gravitational field however, one could take a step back and revert the viewpoint to ask how the presence of a (possibly quantum) gravitational field could affect quantum theory itself. This is sometimes referred to as the attempt to ``gravitize the quantum'' \cite{Penrose_2014,Freidel_2014,Howl_2019}. Likewise, it is natural to ask which structures of quantum theory could be modified in order to capture certain features of possible theories of quantum gravity or quantum spacetime, and which freedom there is to do so.

In LQG for example, one is forced to work with a non-separable (kinematical) Hilbert space spanned by spin network states. This provides a representation of the Weyl algebra of exponentiated connection and conjugated flux operators in which only holonomies of the connection exist as operators but not the connection itself \cite{Lewandowski_2006,Fleischhack_2008}. Transposing this verbatim to the case of a finite-dimensional system leads to the so-called polymer representation of quantum mechanics \cite{HALVORSON200445,Strocchi_2016,Ashtekar:2002sn,Corichi:2007aa}. This representation is \textit{inequivalent} to the standard Schr\"odinger representation in that it only admits an exponentiated momentum operator; the momentum operator itself (and its action as a derivative operator) does not exist, as one could indeed expect from the absence of a spatial continuum at the fundamental level. Such a representation is said to be \emph{irregular}, in the sense that it is not weakly continuous in the exponentiation parameter and thus it evades the Stone--von Neumann theorem. When applied to symmetry-reduced models of gravity, describing e.g. an homogeneous and isotropic universe, this polymer quantization leads to loop quantum cosmology (LQC hereafter) \cite{Ashtekar_2011}, a model of quantum cosmology whose key feature is to resolve the Big Bang singularity. In LQC, the key mathematical input inherited from full LQG is the non-standard Hilbert space representation of the finite-dimensional algebra of position and momentum operators. Studying polymer quantum mechanics on finite-dimensional systems is therefore an example of how inputs from quantum gravity can lead to some insights into the structure of quantum theory itself. Furthermore, it is interesting to notice that the polymer representations have found applications in ordinary quantum field theory \cite{10025994495,Acerbi_1993,L_ffelholz_1995,NILL_1992,Hossain:2010eb,Strocchi_2016}, which suggests that they are not simply mathematical curiosities and that they could play a fundamental role in the description of certain physical systems and processes.

Recently, it has also been suggested that any fundamental theory of quantum gravity should posses some controlled notion of non-locality because of the presence of the fundamental Planck length scale \cite{Freidel_2015,Freidel:2016aa,Freidel_2017,Freidel_2019}. This has revived interest in the so-called modular representations of quantum mechanics, inspired by the modular variables of Aharonov and collaborators \cite{Aharonov_1969,2003qpqt.book.....A}. The modular representations also rely on the Weyl algebra of exponentiated position and momentum operators, but at the difference with the polymer representations (which can be irregular in either position or momentum) are regular, and therefore \textit{equivalent} to the usual Schr\"odinger representation. The interesting physical property of the modular representations is that they rely neither strictly on a position nor momentum polarization, but instead on a choice of modular cell in phase space, thereby interpolating between the two usual polarizations. This is possible provided one has access to a length scale. Such modular representations were actually introduced and used in the context of condensed matter physics \cite{Zak:1967aa,Zak:1993aa}, but recent work has suggested that they could be much more fundamental and actually underlie the notion of quantum space \cite{Freidel:2016aa}. This possibility has triggered a new interest and the systematic investigation of these modular representation in quantum mechanics and field theory. The companion paper \cite{Yargic:2020kkb} studies the formulation of a path integral in modular space. Interestingly, the duality between position and momentum highlighted by the modular representation is also the central theme of so-called Born reciprocity \cite{FREIDEL2014302}, whose implementation in the general relativistic setup is a fascinating open question studied in the field of Born geometry \cite{Svoboda:2019fpt}. The modular representation can be understood in a geometrical manner using Born geometry.

The polymer and modular representations, even when considered in the finite-dimensional case, are examples of alternatives to the usual Schr\"odinger representation motivated by some extra inputs: the non-existence of either position or momentum in the polymer case, and the duality between position and momentum in the modular case. A natural question is therefore whether these two ideas can be combined in order to obtain irregular modular representations which treat position and momentum on the same footing. This is what we achieve in this paper. More precisely, we construct a new family of representations of the Weyl algebra, which we call the ``modular polymer (MP) representations''. These MP representations are \textit{inequivalent} to the Schr\"odinger representation and to the $q$-polymer or $p$-polymer representations. They can be thought of as modular representations in which neither the position nor the momentum operators exist. Since our construction is based on the idea of ``polymerizing'' the modular representations, we actually obtain a family of MP representations indexed by a so-called modular lattice, with each such lattice defining a different MP representation of the Weyl algebra. We however exhibit a unitary isomorphism between these representations, thereby providing a uniqueness result for the MP representations.

Having defined the MP representations, we probe their physical consequence by studying the example of the dynamics of a harmonic oscillator. This is the typical simplest system on which the polymer representations have been investigated before \cite{Corichi:2007aa,Corichi:2006qf,G.:2013lia}. Because of the presence of position and momentum operators in the Hamiltonian, working with the $q$-polymer or $p$-polymer representations requires to approximate either operator by their corresponding Weyl operator in the quantum theory. The MP representations require to exponentiate both. We explain how the treatment of the harmonic oscillator in the MP representation requires to introduce two coarse-graining length and momentum scales. The main result coming from this study is the form of the energy spectrum, which is bounded from both below and above. Depending on the ratio of the modular and coarse-graining scales, the spectrum can be either continuous or discrete.

The construction presented in this work opens up the possibility of studying the quantization of gravitational systems in the MP representation. We expect that the natural presence of pairs of (dual) scales in the MP representation will enable us to have not only a notion of fundamental length (as in the LQG-inspired polymer quantization), but also fundamental energy. In the finite-dimensional case, this will lead to an inequivalent representation of quantum cosmology, for example, which will serve as a toy model to understand features of the full theory with infinitely-many degrees of freedom (such as the BF representation of LQG). In this context, the MP representation will potentially shed light onto the implementation of Born reciprocity in quantum gravity.

This paper is organized as follows. In section \ref{sec:2} we recall the definition of the Weyl algebra in one dimension, and introduce the Sch\"odinger, modular, and polymer representations. In section \ref{sec:3} we then construct the new irregular MP representations by ``polymerizing'' the modular representations. In section \ref{sec:4} we establish the unitary equivalence between the various choices of MP representations. Section \ref{sec:5} is then devoted to the study of the dynamics of the harmonic oscillator in the MP representation, in order to highlight the new physical aspects introduced by this latter. Finally, we present perspectives for future work in section \ref{sec:6}.

In order not to interrupt the flow of the paper too often with details of calculations, we gather all the proofs in the appendix, along with other remarks. Whenever this occurs we signal it with a link \{\textcolor{blue}{$n$}\}, which can be clicked to go to the corresponding note in the appendix, and then clicked once again to come back to the main text.

\section{Representations of the Weyl algebra}
\label{sec:2}

For the sake of clarity, we will base our discussion on the simplest possible mechanical system, which is a non-relativistic particle living on the real line $\R$. On the corresponding phase space $\R^2$ we choose local position and momentum coordinates $(q,p)$, with Poisson bracket $\{q,p\}=1$. The first step towards the quantum theory is then to promote this bracket to the Heisenberg commutation relations $[\hat{q},\hat{p}]=i\hbar\hat{\one}$ between position and momentum operators.

Now, given any two real numbers $(a,b)\in\R^2$ with respective dimensions of length and momentum, we define the so-called Weyl operators as the following exponentiated combination of position and momentum operators:
\be
\Weyl{(a,b)} \equiv
e^{i (b \hat{q} - a \hat{p}) / \hbar} .
\ee
The Weyl algebra $\WS$ (sometimes also referred to as the Heisenberg--Weyl algebra) is the non-commutative $C^*$-algebra generated by these Weyl operators, where the involution and the product are given by
\endnote{The product relation follows from the BCH formula. We have
\be
\Weyl{(a,b)} \, \Weyl{(a',b')}
&=  e^{i (b \hat{q} - a \hat{p}) / \hbar
    \,+\, i (b' \hat{q} - a' \hat{p}) / \hbar
    \,-\, \half [b \hat{q} - a \hat{p} ,\,
        b' \hat{q} - a' \hat{p}]/\hbar^2}
\nn\\
&=  e^{\ihalf (ba'-ab') / \hbar} \, \Weyl{(a+a',b+b')}
.
\ee
The exponentiated position and momentum operators often used in the literature are
\be
\hat{U}_{(a)}\equiv\Weyl{(0,a)}=e^{ia\hat{q}/\hbar},
\q\q
\hat{V}_{(a)}\equiv\Weyl{(a,0)}=e^{-ia\hat{p}/\hbar},
\ee
and they obey evidently the relations $\hat{U}_{(a)}^\dagger=\hat{U}_{(-a)}$ and $\hat{V}_{(a)}^\dagger=\hat{V}_{(-a)}$ as well as the product rules
\be
\hat{U}_{(a)}\hat{U}_{(b)}=\hat{U}_{(a+b)},
\q\q
\hat{V}_{(a)}\hat{V}_{(b)}=\hat{V}_{(a+b)},
\q\q
\hat{U}_{(a)}\hat{V}_{(b)}=e^{iab/\hbar}\hat{V}_{(b)}\hat{U}_{(a)}.
\ee}
\be\label{Weyl relations}
\Weyl{(a,b)}^\dagger = \Weyl{(-a,-b)},
\q\q
\Weyl{(a,b)} \, \Weyl{(a',b')}
= e^{\ihalf (ba'-ab') / \hbar} \, \Weyl{(a+a',b+b')}.
\ee
We shall now consider these relations as \textit{defining} the Weyl algebra of abstract operators $\Weyl{(a,b)}$, and thus forget altogether about the Heisenberg commutation relations and the initial reference to position and momentum operators. The task is then to look for representations of this algebra. We will see that this leads to physically interesting (and potentially inequivalent) alternatives to the usual Schr\"odinger representation. These are precisely the modular and polymer representations, together with their modular polymer generalization.

%\endnote{The reader will notice the alternative use of plural or singular for the Schr\"odinger, modular, and polymer representation(s). In the Schrodinger case, the plural denotes all the representations which are equivalent by virtue of the Stone--von Neumann uniqueness theorem, as well as the equivalent choices of position or momentum polarization. In the modular case, the plural denotes the equivalent modular representations labelled by a choice of modular lattice. In the polymer case however, it denotes the inequivalent $q$-polymer and $p$-polymer representations, and the additional freedom of choosing a polarization in each case.}

\subsection{Schr\"odinger representation}

The Stone--von Neumann uniqueness theorem \cite{Stone_1930,v_Neumann_1931} (see \cite{Rosenberg,Summers_2001} for more recent and pedagogical treatments) states that any irreducible representation of $\WS$ which is weakly continuous
\endnote{Irreducibility and weak continuity imply that the Hilbert space of the representation must be separable.}
in both $a$ and $b$ is unitarily equivalent to the usual Schr\"odinger representation. In the position polarization for example, the Schr\"odinger representation is defined on the Hilbert space $\HS_\text{S}=L^2(\R,\dd q)$ and wave-functions in this latter are acted on as
\endnote{The state $\psi(q)$ stands for $\ket{\psi} = \int \dd q \, \psi(q) \ket{q}$, where $\psi(q)$ in the second expression is a scalar function and $\braket{q}{p} = (2\pi\hbar)^{-1/2} \, e^{ipq/\hbar}$. Then,
\be
\Weyl{(a,b)} \, \psi(q) &\simeq
\Weyl{(a,b)} \int \dd q \, \psi(q) \ket{q}
\nn \\
&= e^{-\ihalf  a b / \hbar} \,
e^{i b \hat{q} / \hbar} \, e^{- i a \hat{p} / \hbar}
\int \dd q \, \psi(q) \int \dd p
\pr{2\pi\hbar}^{-1/2} e^{-ipq/\hbar} \ket{p}
\nn \\
&= e^{-\ihalf  a b / \hbar} \,
e^{i b \hat{q} / \hbar}
\int \dd q \, \psi(q) \int \dd p
\pr{2\pi\hbar}^{-1/2} e^{-ip(q+a)/\hbar} \ket{p}
\nn \\
&= e^{-\ihalf  a b / \hbar} \,
e^{i b \hat{q} / \hbar}
\int \dd q \, \psi(q) \ket{q+a}
\nn \\
&= e^{-\ihalf  a b / \hbar}
\int \dd q \, \psi(q) \, e^{i b (q+a) / \hbar} \ket{q+a}
\nn \\
&= e^{-\ihalf  a b / \hbar}
\int \dd q \, \psi(q-a) \, e^{i b q / \hbar} \ket{q}
\nn \\
&= \int \dd q \pr{e^{-\ihalf  a b / \hbar} \, e^{i b q / \hbar} \,
\psi(q-a)} \ket{q}
\nn \\
&\simeq e^{-\ihalf  a b / \hbar} \, e^{i b q / \hbar} \, \psi(q-a)
.
\ee
The last expression stands for the state in the second-to-last expression.
}
\be
\Weyl{(a,b)} \, \psi(q) =
e^{-\ihalf  a b / \hbar} \, e^{ibq} \, \psi(q-a),
\ee
which is indeed a unitary action.
\endnote{By virtue of \eqref{Weyl relations} we have that
\be
\Weyl{(a,b)}^\dagger \, \psi(q)
= \Weyl{(-a,-b)} \, \psi(q)
= e^{-\ihalf ab/\hbar} \, e^{-ibq/\hbar} \, \psi(q+a)
,
\ee
and therefore
\be
\Weyl{(a,b)} \Weyl{(a,b)}^\dagger \psi(q)
= e^{-\ihalf ab/\hbar} \, e^{-ibq/\hbar} \, \Weyl{(a,b)} \, \psi(q+a)
= \psi(q) .
\ee}
The condition of weak continuity in $a$ and $b$ ensures that there exist self-adjoint operators $\hat{q}$ and $\hat{p}$ on $\HS_\text{S}$ such that $\Weyl{(a,0)} = e^{-ia\hat{p}/\hbar}$ and $\Weyl{(0,b)} = e^{ib\hat{q}/\hbar}$. In other words, the position and momentum operators themselves are well-defined and act as
\be
\hat{q}\psi(q)=q\psi(q),\q\q\hat{p}\psi(q)=-i\hbar\f{\partial}{\partial q}\psi(q).
\ee
This is the rigid structure on which finite-dimensional quantum mechanics relies. However, as we have argued in the introduction, one may want to take a step back and ask if there are representations which capture essential features of possible theories of quantum gravity or quantum spacetime. Two interesting candidates are the modular and polymer representations. The modular representation, although being unitarily equivalent to the Schrödinger one, has the distinguishing feature of relying neither on a strict position nor momentum polarization, but instead on the notion of modular cell in phase space. The polymer representations, on the other hand, are inequivalent to the Schrödinger ones because they bypass the Stone--von Neumann uniqueness theorem by dropping the requirement of existence of either the position or the momentum operator. Let us now explain more precisely how these representations are built.

\subsection{Modular representations}

The essence of the Schr\"odinger representation introduced above is that it singles out the set of operators which are functions of (say) $\hat{q}$ only, and as such commute and can be diagonalized. It can furthermore seem natural to diagonalize $\hat{q}$ since this latter has a natural classical analogue, which here evidently is the position. However, as famously pointed out by Aharonov and collaborators \cite{Aharonov_1969,2003qpqt.book.....A}, there can be situations in which one is forced to work with exponentiated position or momentum operators of the form $\hat{U}_{(2\pi\hbar/\ell)}=e^{2i\pi\hat{q}/\ell}$ and $\hat{V}_{(\ell)}=e^{-i\ell\hat{p}/\hbar}$, where $\ell$ is a length scale. These are examples of modular operators, after which the representation introduced below is named, and they must be used for example if one wants to measure the relative phase in an interference experiment. Remarkably, although classically these functions of position and momentum cannot commute, at the quantum level the operators commute, and as such can be diagonalized simultaneously.

Having this in mind, it is therefore natural to search for the generic condition for a set of exponentiated position and momentum operators to commute. With the Weyl operators at our disposal, this amounts to looking at commutative *-subalgebras of the Weyl algebra. One can see that these are characterized by the condition
\endnote{This follows from the fact that
\be
\Big[\Weyl{(a,b)},\Weyl{(a',b')}\Big]
= \pr{e^{2\pi i(a'b-ab')/(2\pi\hbar)}-1} \Weyl{(a',b')} \, \Weyl{(a,b)}.
\ee}
\be
\Big[\Weyl{(a,b)},\Weyl{(a',b')}\Big]=0
\q\Leftrightarrow\q
\f{1}{2\pi\hbar}(a'b-ab')\in\Z,
\ee
and are therefore in one-to-one correspondence with so-called \textit{modular lattices}.
\endnote{The solution space forms a lattice because the relation $(a'b-ab')/(2\pi\hbar)\in\Z$ is $\Z$-bilinear in both $(a,b)$ and $(a',b')$.}
More precisely, a modular lattice, denoted by $\Lambda$, is defined by a length scale $\ell$ and the corresponding dual momentum scale $\tell \equiv 2\pi\hbar / \ell$ as
\be
\label{ModularLatticeDef}
\Lambda \equiv
	\cur{(a,b) \in \R^2: a \in \ell \, \Z,
		b \in \tell \, \Z}
\,.
\ee
We will denote elements of a modular lattice by $\K$, and omit the explicit reference to the lattice if no confusion can occur.  Naturally, for $\K,\K' \in \Lambda$ the Weyl operators on the modular lattice $\Lambda$ build a commutative *-subalgebra $\WL$ via the relations
\be
\Weyl{\K}^\dagger = \Weyl{-\K}^{\phantom{\dagger}}
\;, \q\q
\Weyl{\K} \, \Weyl{\K'}
= e^{\ihalf  \omega(\K,\K')/\hbar} \, \Weyl{\K+\K'}
\;,
\ee
where we have written the product with the help of the symplectic structure $\omega$ on the phase space $\R^2$, which is defined as
\be\label{symplectic structure}
\omega(\X,\Y) \equiv \tilde{x} y - x \tilde{y}
\ee
for any $\X = (x,\tilde{x}), \Y = (y,\tilde{y}) \in \R^2$. Note that for elements $\K,\K' \in \Lambda$ in a modular lattice we have $e^{\ihalf  \omega(\K,\K')/\hbar} \in \{ +1, -1 \}$.

With this definition, we can now also define the \textit{modular space} (or \textit{modular cell}) $T_\Lambda$ as a torus of size $\ell \times \tell$ which is dual to the modular lattice $\Lambda$, i.e.
\begin{align}
T_\Lambda \equiv \R^2 / \Lambda
\;.
\end{align}
One can see that this modular space has twice the dimension of the polarization space which we would obtain in the Schr\"odinger representation by diagonalizing $\hat{q}$ (or $\hat{p}$ in the momentum polarization). The modular construction therefore admits two singular limits. In the limit $\ell \rightarrow 0$, the modular lattice becomes the position space and the modular space becomes the momentum space. In the limit $\ell \rightarrow \infty$, vice versa. Both cases are degenerate as $\Lambda$ and $T_\Lambda$ change their topology. In addition, there are two trivial cases given by $a\in\R$, $b = 0$ and $a = 0$, $b\in\R$, which correspond to the limits $\ell \rightarrow \infty$ and $\ell \rightarrow 0$, and give rise to the Schr\"odinger representation and its dual, respectively.

We now briefly review the \textit{modular representations} of the Weyl algebra. Given a modular lattice $\Lambda$ and an element $\X = (x,\tilde{x}) \in \R^2$, we define a \textit{modular vector} $\ket{\X}_\Lambda$ in terms of the position eigenstates as
\endnote{It is also possible to write a modular vector similarly as a superposition of momentum eigenstates, such that
\begin{align*}
\ket{\X}_\Lambda
=	\ell^{-1/2} \,
	e^{- i x \tilde{x} / \hbar} \,
	\sum_{\tilde{n} \in \Z}
	e^{- i \tilde{n} \tell x / \hbar}
	\, \vert \tilde{x} + \tilde{n} \tell \rangle
\;.
\end{align*}
Proof:
\begin{align*}
\ket{\X}_\Lambda
&=	\tell^{-1/2}
	\sum_{n\in\Z} e^{i  n \ell  \tilde{x}/\hbar} \, \one
	\ket{x +  n \ell }
\displaybreak[0] \\
&=	\tell^{-1/2}
	\int_\R \dd p 
	\sum_{n\in\Z} e^{i  n \ell  \tilde{x}/\hbar}
	\ket{p}	\braket{p}{x +  n \ell }
\displaybreak[0] \\
&=	\pr{2\pi \hbar / \ell}^{-1/2}
	\int_\R \dd p 
	\sum_{n\in\Z} e^{i  n \ell  \tilde{x}/\hbar}
	\ket{p} \frac{1}{\sqrt{2\pi\hbar}} \,
	e^{- i p (x +  n \ell )/\hbar}
\displaybreak[0] \\
&=	\frac{\sqrt{\ell}}{2\pi\hbar}
	\int_\R \dd p \;
	e^{-ipx/\hbar} \ket{p}
	\sum_{n\in\Z} e^{i  n \ell  (\tilde{x} - p)/\hbar}
\displaybreak[0] \\
&=	\frac{\sqrt{\ell}}{2\pi\hbar}
	\int_\R \dd p \;
	e^{-ipx/\hbar} \ket{p}
	\sum_{\tilde{n}\in\Z}
	\pr{2\pi\hbar / \ell}
	\delta\big(\tilde{x} - p + \tilde{n}\pr{2\pi\hbar / \ell}\big)
\displaybreak[0] \\
&=	\ell^{-1/2} \,
	e^{- i x \tilde{x} / \hbar} \,
	\sum_{\tilde{n} \in \Z}
	e^{- i \tilde{n} \tell x / \hbar}
	\, \vert \tilde{x} + \tilde{n} \tell \rangle
\end{align*}
}
\begin{align}
\ket{\X}_\Lambda &\equiv
	\tell^{-1/2}
	\sum_{n\in\Z} e^{i  n \ell  \tilde{x}/\hbar}
	\ket{x +  n \ell }
\;.
\end{align}
This is known as the Zak transform of the position eigenstates \cite{Zak:1967aa,Zak:1993aa}. Form this definition, one can see that the modular vectors are quasi-periodic. Indeed, for any $\K = (k,\tilde{k}) \in \Lambda$ they satisfy
\endnote{Proof:
\begin{align*}
\ket{\X + \K}_\Lambda
&=	\tell^{-1/2}
	\sum_{n\in\Z}
	e^{i  n \ell  (\tilde{x} + \tilde{k})/\hbar}
	\ket{x + k +  n \ell }
\displaybreak[0] \\
&=	\tell^{-1/2}
	\sum_{n\in\Z}
	e^{i  n \ell  \tilde{x}/\hbar}
	\ket{x + k +  n \ell }
\displaybreak[0] \\
&=	\tell^{-1/2}
	\sum_{n\in\Z}
	e^{i ( n \ell  - k) \tilde{x}/\hbar}
	\ket{x +  n \ell }
\displaybreak[0] \\
&=	e^{-ik\tilde{x}/\hbar} \ket{\X}_\Lambda
\end{align*}
}
\begin{align}
\label{QuasiPerModVec}
\ket{\X + \K}_\Lambda
&=	e^{-i k \tilde{x}/\hbar} \ket{\X}_\Lambda
\;.
\end{align}
Furthermore, the modular vectors diagonalize the commutative subalgebra $\WL$, whose operators act as
\begin{align}
\Weyl{\K} \ket{\X}_\Lambda
&=	e^{\ihalf  k \tilde{k}/\hbar} \,
    e^{i (\tilde{k} x - k \tilde{x})/\hbar}
	\ket{\X}_\Lambda
\;.
\end{align}
This property follows from the action of a Weyl operator for a general $(a,b) \in \R^2$, which is given by
\endnote{Proof:
\begin{align*}
\Weyl{(a,b)} \ket{\X}_\Lambda
&=	e^{-\ihalf  a b / \hbar} \,
    e^{i b \hat{q} / \hbar} \, e^{- i a \hat{p} / \hbar} \,
	\ell^{-1/2}
	\sum_{\tilde{n} \in \Z}
	e^{- i (\tilde{x} + \tilde{n} \tell) x / \hbar}
	\, \vert \tilde{x} + \tilde{n} \tell \rangle
\displaybreak[0] \\
&=	e^{-\ihalf  a b / \hbar} \, e^{i b \hat{q} / \hbar} \,
	\ell^{-1/2}
	\sum_{\tilde{n} \in \Z}
	e^{- i (\tilde{x} + \tilde{n} \tell) (x + a) / \hbar}
	\, \vert \tilde{x} + \tilde{n} \tell \rangle
\displaybreak[0] \\
&=	e^{-\ihalf  a b / \hbar} \, e^{i b \hat{q} / \hbar}
	\ket{\X + (a,0)}_\Lambda
\displaybreak[0] \\
&=	e^{-\ihalf  a b / \hbar} \, e^{i b \hat{q} / \hbar} \,
	\tell^{-1/2}
	\sum_{n\in\Z} e^{i  n \ell  \tilde{x}/\hbar}
	\ket{x + a +  n \ell }
\displaybreak[0] \\
&=	e^{-\ihalf  a b / \hbar} \,
	\tell^{-1/2}
	\sum_{n\in\Z} e^{i  n \ell  \tilde{x}/\hbar} \,
	e^{i b (x + a +  n \ell ) / \hbar}
	\ket{x + a +  n \ell }
\displaybreak[0] \\
&=	e^{\ihalf  a b / \hbar} \, e^{i b x / \hbar} \,
	\tell^{-1/2}
	\sum_{n\in\Z} e^{i  n \ell  (\tilde{x} + b)/\hbar}
	\ket{x + a +  n \ell }
\displaybreak[0] \\
&=	e^{\ihalf  a b / \hbar} \, e^{i b x / \hbar}
    \ket{\X + (a,b)}_\Lambda
\;.
\end{align*}
The previous identity, $\Weyl{\K} \ket{\X}_\Lambda = e^{\ihalf  k \tilde{k}/\hbar} \, e^{i (\tilde{k} x - k \tilde{x})/\hbar} \ket{\X}_\Lambda$, follows from this more general identity, together with the quasi-periodicity relation.
}
\begin{align}
\label{WeylOnModVec}
\Weyl{(a,b)} \ket{\X}_\Lambda
&=	e^{\ihalf  ab/\hbar} \,
    e^{i b x /\hbar} \ket{\X + (a,b)}_\Lambda
\;.
\end{align}
The set $\cur{\ket{\X} |\ \X \in T_\Lambda}$ is complete, meaning that the modular vectors provide a decomposition of the identity of the form
\endnote{
Since the combination $\ketbra{\X}{\X}_\Lambda$ is periodic, we can evaluate the right-hand side of this decomposition of the identity on an arbitrary modular cell of size $\ell \times \tell$. We find
\begin{align*}
\int_{T_\Lambda} \dd^2\X \ketbra{\X}{\X}_\Lambda
&=	\int_0^\ell \dd x \int_0^{\tell} \dd \tilde{x}
	\ketbra{\X}{\X}_\Lambda
\displaybreak[0] \\
&=	\tell^{-1}
	\int_0^\ell \dd x \int_0^{\tell} \dd \tilde{x}
	\sum_{n\in\Z} \sum_{n'\in\Z}
	e^{i  n \ell  \tilde{x}/\hbar} \,
	e^{-i  n \ell ' \tilde{x}/\hbar}
	\ketbra{x +  n \ell }{x +  n \ell '}
\displaybreak[0] \\
&=	\tell^{-1}
	\int_0^\ell \dd x \int_0^{\tell} \dd \tilde{x}
	\sum_{n\in\Z} \sum_{m\in\Z}
	e^{i \ell m \tilde{x}/\hbar}
	\ketbra{x +  n \ell }{x + (n-m)\ell}
\displaybreak[0] \\
&=	\tell^{-1}
	\int_0^\ell \dd x
	\sum_{n\in\Z} \sum_{m\in\Z}
	\ketbra{x +  n \ell }{x + (n-m)\ell}
	\int_0^{\tell} \dd \tilde{x} \;
	e^{i \ell m \tilde{x}/\hbar}
\displaybreak[0] \\
&=	\tell^{-1}
	\int_0^\ell \dd x
	\sum_{n\in\Z} \sum_{m\in\Z}
	\ketbra{x +  n \ell }{x + (n-m)\ell}
	\tell \, \delta_{m,0}
\displaybreak[0] \\
&=	\int_0^\ell \dd x \sum_{n\in\Z}
	\ketbra{x +  n \ell }{x +  n \ell }
\displaybreak[0] \\
&=	\int_\R \dd x \ketbra{x}{x}
=	\one
\;.
\end{align*}
}
\be
\one = \int_{T_\Lambda} \dd^2\X \ketbra{\X}{\X}_\Lambda
\;.
\ee
Furthermore, the modular vectors are orthonormal in the inner product given by
\endnote{Proof:
\begin{align*}
\braket{\X'}{\X}_\Lambda
&=	\tell^{-1}
	\sum_{n'\in\Z} \sum_{n\in\Z}
	e^{-i  n \ell ' \tilde{x}'/\hbar} \,
	e^{i  n \ell  \tilde{x}/\hbar}
	\braket{x' +  n \ell '}{x +  n \ell }
\displaybreak[0] \\
&=	\tell^{-1}
	\sum_{n'\in\Z} \sum_{n\in\Z}
	e^{-i  n \ell ' \tilde{x}'/\hbar} \,
	e^{i  n \ell  \tilde{x}/\hbar} \,
	\delta\big(x' - x + (n' - n)\ell\big)
\displaybreak[0] \\
&=	\tell^{-1}
	\sum_{n'\in\Z} \sum_{n\in\Z}
	e^{-i  n \ell ' \tilde{x}'/\hbar} \,
	e^{i \ell (n' - n) \tilde{x}/\hbar} \,
	\delta(x' - x +  n \ell )
\displaybreak[0] \\
&=	\tell^{-1}
	\sum_{n\in\Z}
	e^{- i  n \ell  \tilde{x}/\hbar} \,
	\delta(x' - x +  n \ell )
	\sum_{n'\in\Z}
	e^{-i  n \ell ' (\tilde{x}' - \tilde{x}) / \hbar}
\displaybreak[0] \\
&=	\tell^{-1}
	\sum_{n\in\Z}
	e^{- i  n \ell  \tilde{x}/\hbar} \,
	\delta(x' - x +  n \ell )
	\sum_{\tilde{n}\in\Z} \tell \,
	\delta(\tilde{x}' - \tilde{x} + \tilde{n} \tell)
\displaybreak[0] \\
&=	\sum_{\K \in \Lambda} e^{-ik\tilde{x}/\hbar} \,
	\delta^2(\X' - \X + \K)
\;.
\end{align*}
The set $\cur{\ket{\X} |\ \X \in T_\Lambda}$ is thus ``orthonormal'' in the sense that $\braket{\X'}{\X}_\Lambda = \delta^2(\X' - \X)$ for $\X,\X' \in T_\Lambda$.}
\be
\braket{\X'}{\X}_\Lambda
=	\sum_{\K \in \Lambda} e^{-ik\tilde{x}/\hbar} \,
	\delta^2(\X' - \X + \K)
\;.
\ee
Finally, we have that the Hilbert space $\mathcal{H}_\Lambda$ consists of vectors of the form
\begin{align}
\ket{\psi}_\Lambda
&=	\int_{T_\Lambda} \dd^2\X \;
	\psi(\X) \ket{\X}_\Lambda
\end{align}
which satisfy two conditions: 1) square-integrability, and 2) that $\psi(\X) \ket{\X}_\Lambda$ is periodic. Note that the periodicity of the integrand is necessary for the integrals to be well-defined on the torus $T_\Lambda$. Hence, the ``modular wave functions'' $\psi(\X)$ satisfy for all $\X \in \R^2$ and $\K \in \Lambda$ the quasi-periodicity relation
\begin{align}
\psi(\X + \K) = e^{i k \tilde{x} / \hbar} \, \psi(\X)
\;.
\end{align}
In other words, a ``modular wave function'' $\psi$ is an $L^2$-section of the $U(1)$-bundle $E_\Lambda \rightarrow T_\Lambda$ over the modular space, which is defined by the identification
\begin{align}
\label{Twist}
(\theta,\X) \sim
\big(\theta \, e^{-ik\tilde{x}/\hbar}, \X + \K\big)
\end{align}
for any $\X = (x,\tilde{x}) \in \R^2$, $\K = (k,\tilde{k}) \in \Lambda$ and $\theta \in U(1)$. We can therefore simply write the Hilbert space for the modular representation as $\mathcal{H}_\Lambda = L^2(E_\Lambda)$.

Having introduced the modular representations, it is now important to point out that they are regular, in the sense that the map $(a,b) \mapsto \Weyl{(a,b)}$ is weakly continuous. By the Stone--von Neumann theorem, this implies that the modular representations (labelled by a choice of modular lattice) are unitarily equivalent to the Schr\"odinger representation and also to each other.

Finally, using this regularity assumption, it is interesting to note that the position and momentum operators are represented in a modular representation as $\hat{q} \sim q + i\hbar \, \pd_p$ and $\hat{p} \sim - i \hbar \, \pd_q$, and in particular that the position operator acquires a shift.
\endnote{Proof:
\begin{align*}
\hat{q} \ket{\psi}_\Lambda
&=	\hat{q} \int_{T_\Lambda} \dd^2\X \;
	\psi(\X) \ket{\X}_\Lambda
\displaybreak[0] \\
&=	\hat{q} \int_{T_\Lambda} \dd^2\X \;
	\psi(\X) \, \tell^{-1/2}
	\sum_{n\in\Z} e^{i  n \ell  \tilde{x}/\hbar}
	\ket{x +  n \ell }
\displaybreak[0] \\
&=	\int_{T_\Lambda} \dd^2\X \;
	\psi(\X) \, \tell^{-1/2}
	\sum_{n\in\Z} e^{i  n \ell  \tilde{x}/\hbar}
	\pr{x +  n \ell } \ket{x +  n \ell }
\displaybreak[0] \\
&=	\int_{T_\Lambda} \dd^2\X \;
	\psi(\X) \, \tell^{-1/2}
	\pr{x - i \hbar \, \pd_{\tilde{x}}}
	\sum_{n\in\Z} e^{i  n \ell  \tilde{x}/\hbar}
	\ket{x +  n \ell }
\displaybreak[0] \\
&=	\int_{T_\Lambda} \dd^2\X
	\pr{x \, \psi(\X) + i\hbar \, \pd_{\tilde{x}} \psi(\X)}
	\ket{\X}_\Lambda
\end{align*}
and
\begin{align*}
\hat{p} \ket{\psi}_\Lambda
&=	\hat{p} \int_{T_\Lambda} \dd^2\X \;
	\psi(\X) \ket{\X}_\Lambda
\displaybreak[0] \\
&=	\hat{p} \int_{T_\Lambda} \dd^2\X \;
	\psi(\X) \, \ell^{-1/2}
	\sum_{\tilde{n} \in \Z}
	e^{- i (\tilde{x} + \tilde{n} \tell) x / \hbar}
	\, \vert \tilde{x} + \tilde{n} \tell \rangle
\displaybreak[0] \\
&=	\int_{T_\Lambda} \dd^2\X \;
	\psi(\X) \, \ell^{-1/2}
	\sum_{\tilde{n} \in \Z}
	e^{- i (\tilde{x} + \tilde{n} \tell) x / \hbar}
	\, (\tilde{x} + \tilde{n} \tell)
	\, \vert \tilde{x} + \tilde{n} \tell \rangle
\displaybreak[0] \\
&=	\int_{T_\Lambda} \dd^2\X \;
	\psi(\X) \, \ell^{-1/2}
	\pr{i\hbar \, \pd_x} \sum_{\tilde{n} \in \Z}
	e^{- i (\tilde{x} + \tilde{n} \tell) x / \hbar}
	\, \vert \tilde{x} + \tilde{n} \tell \rangle
\displaybreak[0] \\
&=	\int_{T_\Lambda} \dd^2\X
	\pr{-i\hbar \, \pd_x \psi(\X)}
	\ket{\X}_\Lambda
\;.
\end{align*}
We also show that the operators $\hat{q}$ and $\hat{p}$ preserve the section condition \eqref{Twist}: For any $\K = (k,\tilde{k}) \in \Lambda$, we have
\begin{align*}
(\hat{q}\psi)(\X + \K)
&=	(x + k) \, \psi(\X + \K)
	+ i\hbar \, \pd_{\tilde{x}} \psi(\X + \K)
\\
&=	(x + k) \, e^{ik\tilde{x}/\hbar} \, \psi(\X)
	+ i\hbar \, \pd_{\tilde{x}}
		\pr{e^{ik\tilde{x}/\hbar} \, \psi(\X)}
=	e^{ik\tilde{x}/\hbar} \, (\hat{q}\psi)(\X)
\end{align*}
and
\begin{align*}
(\hat{p}\psi)(\X + \K)
=	-i\hbar \, \pd_x \psi(\X + \K)
=	-i\hbar \, \pd_x
	\pr{e^{ik\tilde{x}/\hbar} \, \psi(\X)}
=	e^{ik\tilde{x}/\hbar} \, (\hat{p}\psi)(\X)
\;.
\end{align*}}

\subsection{Polymer representations}

Inequivalent representations to the Schr\"odinger (or more generally the modular) one can be obtained by relaxing one or several of the assumptions of the Stone--von Neumann uniqueness theorem. The polymer representations are obtained by relaxing the condition of regularity, i.e. of weak continuity of the Weyl operators $\Weyl{(a,b)}$ in either $a$ or $b$ \cite{Strocchi_2016}. More precisely, relaxing the weak continuity in $a$ produces the $p$-polymer representation, while relaxing the weak continuity in $b$ results in the $q$-polymer representation \cite{Corichi:2007aa}. These two polymer representations are inequivalent to each other and to the Schr\"odinger representation. In addition, for each inequivalent polymer representation one has the freedom of working with the equivalent position or momentum polarizations.

For the sake of definiteness, let us focus on the $p$-polymer representation in the position polarization. This is the polymer representation inspired by LQG \cite{Ashtekar:2002sn}. The lack of weak continuity implies that there exists no self-adjoint operator $\hat{p}$ such that $\Weyl{(a,0)} = e^{-ia\hat{p}/\hbar}$. The non-existence of the operator $\hat{p}=-i\hbar\partial/\partial q$ is what one could naturally expect from theories of quantum gravity such as LQG, where the notion of continuum space breaks down (in the sense at least that it is not fundamental but rather emergent).

In this polymer representation, the Weyl operators act in the same way as in the Schr\"odinger representation, i.e.
\be
\Weyl{(a,b)} \, f(q) =
e^{-\ihalf  a b / \hbar} \, e^{ibq} \, f(q-a).
\ee
The Hilbert space $\cal{H}_\text{P}$, however, consists of functions $f(q)$ on $\R$ which vanish away from a countable subset $S_f$, and are square summable in the sense
\be
\sum_{q\in S_f}|f(q)|^2<\infty,
\ee
with an inner product given by
\be
\ang{f,g} = \sum_{q \in S_f \cap S_g} f(q)^* g(q).
\ee
This non-separable Hilbert space is also sometimes denoted by $\mathcal{H}_\text{P}=L^2(\R_\text{d},\dd\mu_\text{d})$, where $\R_\text{d}$ is the real line equipped with the discrete topology, and $\dd\mu_\text{d}$ is the associated discrete measure. In this polymer representation, another important feature is that the position operator actually possesses a complete set of \textit{normalizable} eigenvectors $\cur{\varphi_{q} |\ q \in \R}$ such that
\be
\Weyl{(0,b)} \varphi_{q}=e^{ibq}\varphi_q,\q\q\Weyl{(a,0)} \varphi_{q}=\varphi_{q+a}, \q\q \ang{\varphi_{q},\varphi_{q'}}=\delta_{q,q'},
\ee
where the right-hand side is the Kronecker delta. From this, it is indeed straightforward to see that $\Weyl{(0,b)}$ is weakly continuous, so that there exists a self-adjoint operator $\hat{q}$ such that $\hat{q}\varphi_q=q\varphi_q$. However, we have
\be
\lim_{a\rightarrow0} \ang{\varphi_{q},\Weyl{(a,0)}\varphi_{q}}=0,
\ee
whereas $\Weyl{(0,0)}=1$ and $\ang{\varphi_q,\varphi_q}=1$. This means precisely that $\Weyl{(a,0)}$ fails to be weakly continuous in $a$, and therefore that $\hat{p}$ itself does not exist.

This polymer representation can also be obtained from a Gel’fand--Naimark--Segal construction using a positive linear functional on $\WS$ \cite{Ashtekar:2002sn}. Furthermore, notice that we have presented here a polymer representation which is irregular in $a$ and written in the position polarization. However, we can also change polarization, and also for both polarizations consider the polymer representation which is irregular in $b$ instead. These possibilities are presented and studied at length in \cite{Corichi:2007aa}, together with the corresponding GNS constructions. Finally, let us point out that in \cite{Cavallaro1999} the authors have given a Stone--von Neumann uniqueness theorem for the irregular polymer representations.

Now that we have reviewed the known Schr\"odinger, modular, and polymer representations of the Weyl algebra, we turn to the new result of this work and introduce the modular polymer representations.

\section{Modular polymer representations}
\label{sec:3}

In this section, we ``polymerize'' the modular representation and obtain a new family of irregular representations of the Weyl algebra $\mathcal{W}$, which we call the ``modular polymer (MP) representations''. From now on, we consider the Weyl algebra as an abstract object without reference to the position and momentum operators.

Let $\Lambda$ be a modular lattice as defined in \eqref{ModularLatticeDef}. We denote by $\mathcal{H}^\Lambda_\mathrm{MP} = l^2(E_\Lambda)$ the non-separable Hilbert space of square-summable sections of the $U(1)$-bundle $E_\Lambda \rightarrow T_\Lambda$. This means that each element $f \in \mathcal{H}^\Lambda_\mathrm{MP}$ is supported on a countable subset $S_f$ of $T_\Lambda$ and satisfies $\vvr{f} = \sum_{\X \in S_f} \vr{f(\X)}^2 < \infty$. Moreover, the domain of each element $f \in \mathcal{H}^\Lambda_\mathrm{MP}$ can be extended from $T_\Lambda$ to $\R^2$ by the section condition \eqref{Twist}, i.e.~such that $f(\X+\K) = e^{ik\tilde{x}/\hbar} f(\X)$ for any $\X = (x,\tilde{x}) \in \R^2$ and $\K = (k,\tilde{k}) \in \Lambda$. On the Hilbert space $\mathcal{H}^\Lambda_\mathrm{MP}$ we consider the inner product
\begin{align}
\ang{f,g} = \sum_{\X \in S_f \cap S_g} f(\X)^* g(\X)
\;.
\end{align}
For each $\X \in \R^2$, we then define the function $\varphi^{(\Lambda)}_\X : \R^2 \rightarrow \C$ as
\begin{align}
\label{MP-BasisVectors}
\varphi^{(\Lambda)}_\X(\Y) =
\begin{cases}
e^{i(y-x)\tilde{x}/\hbar} \;,
&\text{if } \Y-\X \in \Lambda
\;, \\
0 \;, &\text{otherwise}
\;,
\end{cases}
\end{align}
for any $\Y = (y,\tilde{y}) \in \R^2$. We will drop the superscript $(\Lambda)$ from $\varphi^{(\Lambda)}_\X$ unless there is an ambiguity about which modular lattice we refer to. Note that the functions $\varphi_\X$ are supported on a single point $\X \in T_\Lambda$ in the modular space. Moreover, they satisfy $\varphi_\X(\Y + \K) = e^{ik\tilde{y}/\hbar} \, \varphi_\X(\Y)$ for any $\K \in \Lambda$.
\endnote{We show $\varphi_\X(\Y + \K) = e^{ik\tilde{y}/\hbar} \, \varphi_\X(\Y)$ as follows:
\begin{align*}
\varphi_\X(\Y + \K)
&=	\begin{cases}
	e^{i(y+k-x)\tilde{x}/\hbar} \;,
	&\text{if } \Y+\K-\X \in\Lambda
	\;, \\
	0 \;, &\text{otherwise}
	\;,
	\end{cases}
\displaybreak[0] \\
&=	\begin{cases}
	e^{ik\tilde{x}/\hbar} \,
	e^{i(y-x)\tilde{x}/\hbar} \;,
	&\text{if } \Y-\X \in\Lambda
	\;, \\
	0 \;, &\text{otherwise}
	\;,
	\end{cases}
\displaybreak[0] \\
&=	\begin{cases}
	e^{ik\tilde{y}/\hbar} \,
	e^{i(y-x)\tilde{x}/\hbar} \;,
	&\text{if } \Y-\X \in\Lambda
	\;, \\
	0 \;, &\text{otherwise}
	\;,
	\end{cases}
\displaybreak[0] \\
&=	e^{ik\tilde{y}/\hbar} \, \varphi_\X(\Y)
\;.
\end{align*}}
\endnote{Another useful property of the functions $\varphi_\X$ is that for any $(a,b) \in \R^2$,
\begin{align}
\label{useful1}
\varphi_\X(\Y+(a,b))
=	e^{i (y+a-x) b/\hbar} \,
	\varphi_{\X-(a,b)}(\Y)
\;.
\end{align}
Proof:
\begin{align*}
\varphi_\X(\Y+(a,b))
&=	\begin{cases}
	e^{i(y+a-x)\tilde{x}/\hbar} \;,
	&\text{if } \Y + (a,b) - \X \in\Lambda
	\;, \\
	0 \;, &\text{otherwise}
	\;,
	\end{cases}
\displaybreak[0] \\
&=	\begin{cases}
	e^{i(y+a-x)b/\hbar} \,
	e^{i(y-(x-a))(\tilde{x}-b)/\hbar} \;,
	&\text{if } \Y - \pr{\X - (a,b)} \in\Lambda
	\;, \\
	0 \;, &\text{otherwise}
	\;,
	\end{cases}
\displaybreak[0] \\
&=	e^{i(y+a-x)b/\hbar} \, \varphi_{\X-(a,b)}(\Y)
\;.
\end{align*}}
Hence, we have $\varphi_\X \in \mathcal{H}^\Lambda_\mathrm{MP}$ for every $\X \in \R^2$. Finally, for any $\K = (k,\tilde{k}) \in \Lambda$, we also have the identity
\endnote{Proof:
\begin{align*}
\varphi_{\X + \K}(\Y)
&=	\begin{cases}
	e^{i(y-x-k)(\tilde{x} + \tilde{k})/\hbar} \;,
	&\text{if } \Y - \pr{\X+\K} \in\Lambda
	\;, \\
	0 \;, &\text{otherwise}
	\;,
	\end{cases}
\displaybreak[0] \\
&=	\begin{cases}
	e^{i(y-x-k)\tilde{k}/\hbar} \,
	e^{-ik\tilde{x}/\hbar} \,
	e^{i(y-x)\tilde{x}/\hbar} \;,
	&\text{if } \Y - \X \in\Lambda
	\;, \\
	0 \;, &\text{otherwise}
	\;,
	\end{cases}
\displaybreak[0] \\
&=	e^{-ik\tilde{x}/\hbar} \, \varphi_{\X}(\Y)
\;.
\end{align*}}
\begin{align}
\label{QuasiPerVarphi}
\varphi_{\X + \K} = e^{-ik\tilde{x}/\hbar} \, \varphi_\X
\;.
\end{align}
We therefore conclude that the set $\cur{\varphi_\X : \X \in T_\Lambda}$ is an orthonormal basis for $\mathcal{H}^\Lambda_\mathrm{MP}$.

It is instructive to compare the functions $\varphi_\X^{(\Lambda)} \in \mathcal{H}^\Lambda_\mathrm{MP}$ to the modular vectors $\ket{\X}_\Lambda$. For example, one can see that identity \eqref{QuasiPerVarphi} is analogous to the relation \eqref{QuasiPerModVec}. In what follows, our strategy will be to push this analogy forward. For each $(a,b) \in \R^2$, we define the action of $\Weyl{(a,b)}$ on the set $\cur{\varphi_\X : \X \in \R^2}$ by
\begin{align}
\label{WeylOnVarphi}
\Weyl{(a,b)} \, \varphi_\X
&=	e^{\ihalf  ab/\hbar} \,
    e^{i b x /\hbar} \, \varphi_{\X + (a,b)}
\;.
\end{align}
The motivation for this definition comes from the equation \eqref{WeylOnModVec}. A straightforward calculation shows that
\endnote{Proof:
\begin{align*}
\Weyl{(a,b)} \Weyl{(a',b')} \, \varphi_\X
&=	e^{\ihalf  a' b'/\hbar} \, e^{i b' x /\hbar} \,
    \Weyl{(a,b)} \, \varphi_{\X + (a',b')}
\displaybreak[0] \\
&=	e^{\ihalf  a' b'/\hbar} \, e^{i b' x /\hbar} \,
    e^{\ihalf  a b/\hbar} \, e^{i b (x+a') /\hbar} \,
    \varphi_{\X + (a,b) + (a',b')}
\displaybreak[0] \\
&=	e^{\ihalf  (b a' - a b') /\hbar} \,
    e^{\ihalf  (a+a') (b+b')/\hbar} \,
    e^{i (b+b') x /\hbar} \,
    \varphi_{\X + (a+a',b+b')}
\displaybreak[0] \\
&=	e^{\ihalf  (b a' - a b') /\hbar} \,
    \Weyl{(a+a',b+b')} \, \varphi_{\X}
\;.
\end{align*}}
\be
\Weyl{(a,b)} \Weyl{(a',b')} \, \varphi_\X
&=	e^{\ihalf  (b a' - a b') /\hbar} \,
    \Weyl{(a+a',b+b')} \, \varphi_{\X}
\;,
\ee
and
\endnote{Proof:
\begin{align*}
\ang{\Weyl{(a,b)} \, \varphi_\X, \varphi_{\X'}}
&=	\sum_\Y \pr{(\Weyl{(a,b)} \, \varphi_\X)(\Y)}^*
	\varphi_{\X'}(\Y)
\displaybreak[0] \\
&=	\sum_\Y \pr{
    e^{\ihalf  a b/\hbar} \, e^{i b x /\hbar} \,
    \varphi_{\X + (a,b)}(\Y)}^*
	\varphi_{\X'}(\Y)
\displaybreak[0] \\
&\hspace{-5.5pt}\overset{\eqref{useful1}}{=}
    \sum_\Y \pr{
    e^{- \ihalf  a b/\hbar} \, e^{i b y/\hbar} \,
    \varphi_{\X}(\Y - (a,b))}^*
	\varphi_{\X'}(\Y)
\displaybreak[0] \\
&=  \sum_\Y \pr{
    e^{\ihalf  a b/\hbar} \, e^{i b y/\hbar} \,
    \varphi_{\X}(\Y)}^*
	\varphi_{\X'}(\Y + (a,b))
\displaybreak[0] \\
&=  \sum_\Y \varphi_{\X}(\Y)^* \,
    e^{- \ihalf  a b/\hbar} \, e^{-i b y/\hbar} \,
	\varphi_{\X'}(\Y + (a,b))
\displaybreak[0] \\
&\hspace{-5.5pt}\overset{\eqref{useful1}}{=}
    \sum_\Y \varphi_{\X}(\Y)^* \, e^{\ihalf  a b/\hbar} \,
	e^{- i x' b/\hbar} \, \varphi_{\X'-(a,b)}(\Y)
\displaybreak[0] \\
&=  \sum_\Y \varphi_{\X}(\Y)^* \, \Weyl{(-a,-b)} \, \varphi_{\X'}(\Y)
\displaybreak[0] \\
&=  \ang{\varphi_\X, \Weyl{(-a,-b)} \, \varphi_{\X'}}
	\;.
\end{align*}}
\be
\ang{\Weyl{(a,b)} \, \varphi_\X, \varphi_\Y}
&=	\ang{\varphi_\X, \Weyl{(-a,-b)} \, \varphi_\Y}
\;.
\ee
Hence, we confirmed that the action of the Weyl operators on $\mathcal{H}^\Lambda_\mathrm{MP}$ as defined in \eqref{WeylOnVarphi} builds a representation of the Weyl algebra. We note that this representation diagonalizes the subalgebra $\WL$, since for any $\K \in \Lambda$,
\begin{align}
\Weyl{\K} \, \varphi_\X
&=	e^{\ihalf  k \tilde{k} / \hbar} \,
    e^{i(\tilde{k} x - k \tilde{x})/\hbar} \, \varphi_\X
\;.
\end{align}
Finally, from
\begin{align}
\ang{\varphi_\X, \Weyl{(a,b)} \, \varphi_\X} =
\begin{cases}
0 \;, &\text{when } (a,b) \notin \Lambda \;, \\
e^{\ihalf  a b / \hbar} \,
e^{i(b x - a \tilde{x})/\hbar}
\;, &\text{when } (a,b) \in \Lambda \;,
\end{cases}
\end{align}
it is clear that $(a,b) \mapsto \Weyl{(a,b)}$ is not weakly continuous. Therefore, there are no self-adjoint operators $\hat{q}$ and $\hat{p}$ such that $\Weyl{(a,b)} = e^{i (b \hat{q} - a \hat{p}) / \hbar}$ for all $(a,b) \in \R^2$.

\section{Generalization and equivalence}
\label{sec:4}

We have so far introduced the MP representations on the example of a 1-dimensional configuration space. In this section, we present a generalization of the construction to arbitrary $d$ dimensions.

For this, we first generalize the symplectic structure introduced in \eqref{symplectic structure} to the $d$-dimensional case by writing it as $\omega(\X,\Y) = \tilde{x} \cdot y - x \cdot \tilde{y}$, and also introduce a metric $\eta(\X,\Y) = \tilde{x} \cdot y + x \cdot \tilde{y}$ on the phase space $\PS=\R^{2d}$. The symplectic structure and the metric endow $\PS$ with a \emph{para-Hermitian structure} \cite{Freidel:2017yuv}. The projectors $ \pr{1 \pm \eta^{-1} \omega}/2$ can be used to distinguish the position and momentum subspaces of $\PS$. The symmetry group that preserves both $\omega$ and $\eta$ is $\mathrm{Sp}(2d) \cap \mathrm{O}(d,d) = \mathrm{GL}(d)$.

We define a modular lattice $\Lambda \subset \PS$ as a maximal discrete set that satisfies $\omega(\Lambda,\Lambda) = 2\pi\hbar \, \Z$. The definition \eqref{MP-BasisVectors} of the MP basis vectors is generalized in $d$ dimensions to
\begin{align}
\varphi^{(\Lambda)}_\X(\Y) =
\begin{cases}
e^{\ihalf (\eta + \omega)(\X,\Y-\X) /\hbar} \;,
&\text{if } \Y-\X \in \Lambda
\;, \\
0 \;, &\text{otherwise}
\;,
\end{cases}
\end{align}
for any $\X,\Y \in \PS$. These functions satisfy
\begin{align}
\label{Gen-WeylOnVarphiMulti}
    \Weyl{\Y} \, \varphi_\X &=
    e^{\frac{i}{4} \eta(\Y,\Y)/\hbar} \,
    e^{\ihalf (\eta + \omega)(\Y,\X) /\hbar} \,
    \varphi_{\X + \Y}
\end{align}
and
\begin{align}
    \varphi_{\X + \K} &=
    e^{\ihalf (\omega - \eta)(\K,\X)/\hbar} \,
    \varphi_\X
\end{align}
for any $\X,\Y \in \PS$, $\K \in \Lambda$.

Since each modular lattice defines a different MP representation of the Weyl algebra, it is natural to ask whether these representations are unitarily equivalent to each other. We now show that this is indeed the case. This question is also tied to the fact that a modular lattice does not actually break the isotropy of space (in spite of being a lattice).

\begin{propo}
Let $\Lambda$ and $\Lambda'$ be two modular lattices, which are related to each other by a matrix $M \in \mathrm{GL}(2d)$ such that $\Lambda' = M \Lambda$. The MP representations based on $\Lambda$ and $\Lambda'$ are unitarily equivalent if and only if $M^T \omega M = \omega$ and $M^T \eta M = \eta$, i.e.~$M \in \mathrm{Sp}(2d) \cap \mathrm{O}(d,d) = \mathrm{GL}(d)$.
\end{propo}
\begin{proof}
Let $M$ be such that $M^T \omega M = \omega$ and $M^T \eta M = \eta$. We define a map $\mathcal{U}_M : \mathcal{H}^\Lambda_\mathrm{MP} \rightarrow \mathcal{H}^{\Lambda'}_\mathrm{MP}$ by
\begin{align}
\mathcal{U}_M \varphi_\X^{(\Lambda)}
=	\varphi_{M\X}^{(\Lambda')}
\;.
\end{align}
This map is unitary because it has an inverse given by $\mathcal{U}_M^{-1}: \mathcal{H}^{\Lambda'}_\mathrm{MP} \rightarrow \mathcal{H}^\Lambda_\mathrm{MP}$, $\varphi_\X^{(\Lambda')} \mapsto \varphi_{M^{-1}\X}^{(\Lambda)}$, and it preserves the inner product,
\begin{align}
\ang{\varphi_{M\X}^{(\Lambda')}, \varphi_{M\Y}^{(\Lambda')}}
	_{\mathcal{H}^{\Lambda'}_\mathrm{MP}}
=	\ang{\mathcal{U}_M \varphi_\X^{(\Lambda)},
		\mathcal{U}_M \varphi_\Y^{(\Lambda)}}
	_{\mathcal{H}^{\Lambda'}_\mathrm{MP}}
=	\ang{\varphi_\X^{(\Lambda)}, \varphi_\Y^{(\Lambda)}}
	_{\mathcal{H}^\Lambda_\mathrm{MP}}
\;.
\end{align}
We would like the Weyl operators to transform under $\mathcal{U}_M$ as
\begin{align}
\label{Gen-WeylOpsTransform}
\mathcal{U}_M \Weyl{\X} \, \mathcal{U}_M^\dagger
	&=  \Weyl{M\X}
\;,
\end{align}
for any $\X \in \PS$, where $\Weyl{\X}$ is considered as a bounded operator on $\mathcal{H}^\Lambda_\mathrm{MP}$ and $\Weyl{M\X}$ is considered as a bounded operator on $\mathcal{H}^{\Lambda'}_\mathrm{MP}$. One can check that
\begin{align}
    \Weyl{M \Y} \, \varphi_{M\X}^{(\Lambda')} &=
    e^{\frac{i}{4} \eta(M\Y,M\Y)/\hbar} \,
    e^{\ihalf (\eta + \omega)(M\Y,M\X) /\hbar} \,
	\varphi_{M\X + M\Y}^{(\Lambda')}
	\nonumber \\
	&=  e^{\frac{i}{4} \eta(\Y,\Y)/\hbar} \,
        e^{\ihalf (\eta + \omega)(\Y,\X) /\hbar} \,
	    \mathcal{U}_M \varphi_{\X + \Y}^{(\Lambda)}
	\nonumber \\
	&=  \mathcal{U}_M \Weyl{\Y} \, \varphi_{\X}^{(\Lambda)}
\end{align}
agrees with \eqref{Gen-WeylOpsTransform}. Therefore, \eqref{Gen-WeylOpsTransform} is consistent and $\mathcal{U}_M$ is a unitary isomorphism.

For the second part of the proof, let $\omega' \equiv M^T \omega M$ or $\eta' \equiv M^T \eta M$. We will check whether \eqref{Gen-WeylOpsTransform} can be postulated when $\omega' \neq \omega$ or $\eta' \neq \eta$. Using the same definitions as before, we find
\begin{align}
    \Weyl{M \Y} \, \varphi_{M\X}^{(\Lambda')}
	&=  e^{\frac{i}{4} \eta'(\Y,\Y)/\hbar} \,
        e^{\ihalf (\eta' + \omega')(\Y,\X) /\hbar} \,
	    \mathcal{U}_M \varphi_{\X + \Y}^{(\Lambda)}
	\nonumber \\
	&=  e^{\frac{i}{4} \eta'(\Y,\Y)/\hbar} \,
        e^{\ihalf (\eta' + \omega')(\Y,\X) /\hbar} \,
	    e^{-\frac{i}{4} \eta(\Y,\Y)/\hbar} \,
        e^{-\ihalf (\eta + \omega)(\Y,\X) /\hbar} \,
	    \mathcal{U}_M \Weyl{\Y} \, \varphi_{\X}^{(\Lambda)}
	\;.
\end{align}
For this to be consistent with \eqref{Gen-WeylOpsTransform}, we need
\begin{align}
    \frac{1}{4} \eta'(\Y,\Y)
    + \half (\eta' + \omega')(\Y,\X)
	- \frac{1}{4} \eta(\Y,\Y)
    - \half (\eta + \omega)(\Y,\X)
    \in 2\pi\hbar \, \Z
    \;, \quad \forall \X,\Y \in \PS
    \;.
\end{align}
This is not possible unless $\omega' = \omega$ and $\eta' = \eta$. Without being able to postulate \eqref{Gen-WeylOpsTransform}, $\mathcal{U}_M$ cannot be a unitary isomorphism between $\mathcal{H}^\Lambda_\mathrm{MP}$ and $\mathcal{H}^{\Lambda'}_\mathrm{MP}$.
\end{proof}

We have therefore shown that the MP representations built on different modular lattices are all unitarily equivalent to each other. Picking a given modular lattice for the construction therefore does not break the rotational symmetry of the space. This was pointed out already in the modular case and explained at length in \cite{Freidel:2016aa}. The reason is essentially that the modular lattice is defined on phase space, which is a non-commutative space.

\section{Dynamics of the harmonic oscillator}
\label{sec:5}

A natural next step is now to investigate the physical consequences of our construction. For this, we would like to consider the dynamics of a quantum harmonic oscillator with Hamiltonian operator
\begin{align}
\hat{H}
=	\frac{1}{2m} \, \hat{p}^2
	+ \half \, m \, \omega^2 \, \hat{q}^2
\end{align}
in an MP representation. However, the operators $\hat{p}$ and $\hat{q}$ (let alone their squares) do not exist in an MP representation. Therefore, we have to construct approximants using the Weyl operators.

The usual procedure in the literature on polymer quantization begins with choosing a ``coarse-graining scale''. For the MP representation we need two scales: a length scale $\lambda$ to approximate $\hat{p}$ and a momentum scale $\tambda$ to approximate $\hat{q}$. In fact, we already have a pair of scales $\ell$ and $\tell = 2\pi\hbar/\ell$ naturally available in an MP representation, but it will turn out to be more natural and general to consider new coarse-graining scales.
\endnote{There are multiple reasons for this. First, the modified operators $\hat{q}$ and $\hat{p}$ at the scale $\ell$ and $\tell$ will belong to the subalgebra $\WL$ and therefore commute, which is undesirable. Second, we want to be able to take the limits $\lambda \rightarrow 0$ and $\tambda \rightarrow 0$ independently, but the scales $\ell$ and $\tell$ are constrained to satisfy $\ell \tell = 2\pi \hbar$. Finally, for the sake of generality, we would like to consider an arbitrary scale pair $(\lambda,\tambda)$, which can be taken at the end of the calculations to be equal to $(\ell,\tell)$ if it turns out to be desirable.}

Let us therefore introduce an arbitrary length scale $\lambda$ and an arbitrary momentum scale $\tambda$, in addition to the ``modular scales'' $(\ell,\tell)$. Next, we introduce a lattice $\LS_{\X_0}^{\Lambda,(\lambda,\tambda)} \subset T_\Lambda$ in the modular space for any $\X_0 \in \R^2$ by
\begin{align}
\label{LatticeDef}
\LS_{\X_0}^{\Lambda,(\lambda,\tambda)} \equiv
\cur{\X \in \R^2 :
	\exists \, n,\tilde{n} \in \Z,
	\exists \, \K \in \Lambda :
	\X = \X_0 + (n \lambda, \tilde{n} \tambda) + \K}
	/ \Lambda
\;.
\end{align}
There are then three cases to consider:
\endnote{Let's visualize the definition \eqref{LatticeDef} with an example. We consider a single direction for simplicity.

Case 1 - irrational scale ratio: Let $\ell = 1$ and $\lambda = 1/\sqrt{2}$. Then, the lattice points have coordinates of the form $x_n = x_0 + n/\sqrt{2} \,\, (\mathrm{mod} \, 1)$ for $n\in\Z$. This is an infinite lattice on a compact space (w.r.t.~continuum topology).

Case 2 - rational scale ratio: Let $\ell = 1$ and $\lambda = 2/5$. Then, we get a finite lattice at the points $x_n = x_0 + n/5 \,\, (\mathrm{mod} \, 1)$ for $n = 0,...,4$. This is a finite lattice.}
\begin{enumerate}
\item
Both $\lambda / \ell$ and $\tambda / \tell$ are \textbf{irrational} numbers. In this case, $\LS_{\X_0}^{\Lambda,(\lambda,\tambda)}$ is an infinite set, which is dense in $T_\Lambda$ with respect to the continuum topology.
\item
Both $\lambda / \ell$ and $\tambda / \tell$ are \textbf{rational} numbers. In this case, $\LS_{\X_0}^{\Lambda,(\lambda,\tambda)}$ is a finite set.
\item
A combination of the previous two cases.
\end{enumerate}
These three cases have different physical consequences and we will discuss them separately.

In each case, we consider the elements $\psi \in \mathcal{H}^\Lambda_\mathrm{MP}$ of the MP Hilbert space which are supported on the lattice $\LS_{\X_0}^{\Lambda,(\lambda,\tambda)} \subset T_\Lambda$. These elements belong to a \textit{separable} Hilbert space $\mathcal{H}^{\Lambda, (\lambda,\tambda), \X_0}_\mathrm{MP}$, which is a \textit{superselection sector} of the full MP Hilbert space. The MP Hilbert space can be written as a direct sum
\begin{align}
\mathcal{H}^\Lambda_\mathrm{MP}
&=	\bigoplus_{\X_0 \,\in\,
		T_\Lambda / \LS_{0}^{\Lambda,(\lambda,\tambda)}}
	\mathcal{H}^{\Lambda, (\lambda,\tambda), \X_0}
		_\mathrm{MP}
\end{align}
over the superselection sectors labeled by $\X_0=(x_0,\tilde{x}_0)$. Finally, we define the operators $\widehat{q^2}$ and $\widehat{p^2}$ on each superselection sector by
\endnote{These definitions are motivated by the Taylor expansion of trigonometric functions. For example,
\begin{align*}
\frac{\hbar^2}{\lambda^2}
	\pr{2 - \Weyl{(\lambda,0)} - \Weyl{(-\lambda,0)}}
\text{``\!}&=\text{\!''\,}
	\frac{\hbar^2}{\lambda^2}
	\pr{2 - e^{-i\lambda \hat{p}/\hbar}
		- e^{i\lambda \hat{p}/\hbar}}
\\
&=	\frac{2\hbar^2}{\lambda^2}
	\pr{1 - \cos\pr{\lambda \hat{p}/\hbar}}
\\
&=	\widehat{p^2} + \mathcal{O}((\lambda p/\hbar)^2)
\;.
\end{align*}
We used the symbol $\text{``\!}=\text{\!''}$ to indicate that the corresponding equation is only a formal writing, since the operator $\hat{p}$ does not exist. We also infer from this calculation that the approximation is valid only in the regime $p \ll \hbar / \lambda$.}
\begin{subequations}
\begin{align}
\widehat{q^2_\tambda\,} &\equiv
	\frac{\hbar^2}{\tambda^2}
	\pr{2 - \Weyl{(0,\tambda)} - \Weyl{(0,-\tambda)}}
\;, \\
\widehat{p^2_\lambda} &\equiv
	\frac{\hbar^2}{\lambda^2}
	\pr{2 - \Weyl{(\lambda,0)} - \Weyl{(-\lambda,0)}}
\;.
\end{align}
\end{subequations}
These definitions are based on approximations that are valid in the regimes $q \ll \hbar / \tambda$ and $p \ll \hbar / \lambda$, but we will consider them as fundamental definitions at all scales. Note that the operators $\widehat{q^2_\tambda\,}$ and $\widehat{p^2_\lambda}$ map each superselection sector onto itself. Finally, let us also point out that here we are considering the regularized operators corresponding to the square of position and momentum, and not the operators squared such as $\hat{q}^2$ and $\hat{p}^2$. While this latter choice also leads to a well-defined regularization of the Hamiltonian, we have made the choice which is usually followed in the literature on polymer quantization \cite{Ashtekar:2002sn,G.:2013lia}.

Using the redefined position and momentum operators, we obtain a regularized Hamiltonian of the form
\begin{align}
\label{HamRef}
\hat{H}_{(\lambda,\tambda)} &\equiv
	\frac{\hbar^2}{2m \lambda^2}
	\pr{2 - \Weyl{(\lambda,0)} - \Weyl{(-\lambda,0)}}
	+ \frac{m \omega^2 \hbar^2}{2 \tambda^2}
	\pr{2 - \Weyl{(0,\tambda)} - \Weyl{(0,-\tambda)}}
\;.
\end{align}
In the following, we will make the additional assumption
\begin{align}
\label{AssumptionLambda}
    \lambda \tambda \in 2\pi\hbar \, \Z
    \;,
\end{align}
and analyze the spectrum of the Hamiltonian \eqref{HamRef} in the first two cases of scale ratios. We will motivate this assumption in the following subsection by investigating the solutions under a power-law ansatz, but the generic case $\lambda \tambda \notin 2\pi\hbar \, \Z$ remains open.

Notice also that the third case mentioned above (i.e. the combination of rational and irrational ratios) does not arise under this assumption. Finally, because $\lambda$ and $\tambda$ are related by \eqref{AssumptionLambda}, in particular $\lambda \tambda \geq 2\pi \hbar$, it is not possible to take the limits $\lambda \rightarrow 0$ and $\tambda \rightarrow 0$ simultaneously under this assumption. This issue disappears in the classical limit, where $\hbar\rightarrow0$.

\subsection{Irrational scale ratios}

Let's start with the case when both $\lambda / \ell$ and $\tambda / \tell$ are irrational numbers. An element $\psi \in \mathcal{H}^{\Lambda, (\lambda,\tambda), \X_0}_\mathrm{MP}$ of the superselection sector can be written uniquely
\endnote{We shall prove the uniqueness of the expression \eqref{PsiInA}. Let's define a lattice $\bar{\LS}_{\X_0}^{(\lambda,\tambda)} \subset \R^2$ for any $\X_0 \in \R^2$ as
\begin{align*}
\bar{\LS}_{\X_0}^{(\lambda,\tambda)} \equiv
	\cur{\Y \in \R^2 :
	\exists \, n,\tilde{n} \in \Z :
	\Y = \X_0 + (n \lambda, \tilde{n} \tambda)}
\;.
\end{align*}
Assuming that $\lambda/\ell$ and $\tambda/\tell$ are irrational numbers, there is a bijective mapping between the lattices $\bar{\LS}_{\X_0}^{(\lambda,\tambda)} \subset \R^2$ and $\LS_{\X_0}^{\Lambda,(\lambda,\tambda)} \subset T_\Lambda$ that is given by
\begin{align*}
\Pi : \bar{\LS}_{\X_0}^{(\lambda,\tambda)}
	&\rightarrow \LS_{\X_0}^{\Lambda,(\lambda,\tambda)} \\
\Y &\mapsto \Y \,\, (\mathrm{mod} \, \Lambda)
\;.
\end{align*}
Moreover, the basis elements satisfy the relation \eqref{QuasiPerVarphi}, i.e.~the states $\varphi_\Y$ and $\varphi_{\Pi(\Y)}$ are linearly dependent for any $\Y \in \R^2$. Since $\cur{\varphi_{\Y} : \Y \in \LS_{\X_0}^{\Lambda,(\lambda,\tambda)}}$ is a basis of the superselection sector $\mathcal{H}^{\Lambda, (\lambda,\tambda), \X_0}_\mathrm{MP}$, the set $\cur{\varphi_\Y : \Y \in \bar{\LS}_{\X_0}^{(\lambda,\tambda)}}$ is also a basis. Hence, the expression \eqref{PsiInA} is simply a decomposition of a state in this basis, which is unique.}
as
\begin{align}
\label{PsiInA}
\psi
&=	\sum_{n,\tilde{n} \in \Z}
	A_{n,\tilde{n}} \,
	\varphi_{\X_0 + (n\lambda,\tilde{n}\tambda)}
\;,
\end{align}
where $A_{n,\tilde{n}}$ are complex numbers. From the action of the Hamilton operator \eqref{HamRef} on this state, we get
\begin{align}
\hat{H}_{(\lambda,\tambda)} \, \psi
&=	\sum_{n,\tilde{n} \in \Z} \bigg(
		\frac{\hbar^2}{2m \lambda^2}
		\pr{2 A_{n,\tilde{n}}
			- A_{n-1,\tilde{n}} - A_{n+1,\tilde{n}}
			}
	\nonumber \\ & \hspace{0.5cm}
		+ \frac{m \omega^2 \hbar^2}{2 \tambda^2}
		\pr{2 A_{n,\tilde{n}}
			- e^{i \tambda (x_0 + n \lambda) / \hbar} \,
			A_{n,\tilde{n}-1}
			- e^{-i \tambda (x_0 + n \lambda) / \hbar} \,
			A_{n,\tilde{n}+1}
			}
		\!\! \bigg)
		\, \varphi_{\X_0 + (n\lambda,\tilde{n}\tambda)}\;,
\end{align}
where we used the equation \eqref{WeylOnVarphi}, and shifted the summation variables $n$ and $\tilde{n}$. Hence, if $\psi$ is an eigenvector of the Hamilton operator $\hat{H}_{(\lambda,\tambda)}$ with eigenvalue $E_{(\lambda,\tambda)}$, the coefficients $A_{n,\tilde{n}}$ are required to satisfy the linear recurrence relation
\begin{align}
\label{RecursionEqA}
E_{(\lambda,\tambda)} \, A_{n,\tilde{n}}
&=	\frac{\hbar^2}{2m \lambda^2}
	\pr{2 A_{n,\tilde{n}}
		- A_{n-1,\tilde{n}} - A_{n+1,\tilde{n}}
		}
	\nonumber \\ & \hspace{0.5cm}
	+ \frac{m \omega^2 \hbar^2}{2 \tambda^2}
	\pr{2 A_{n,\tilde{n}}
		- e^{i \tambda (x_0 + n \lambda) / \hbar} \,
		A_{n,\tilde{n}-1}
		- e^{-i \tambda (x_0 + n \lambda) / \hbar} \,
		A_{n,\tilde{n}+1}
		}
\;.
\end{align}
Before we analyze the spectrum of the Hamiltonian under the assumption \eqref{AssumptionLambda}, we are going to motivate this assumption by examining a power-law ansatz. Consider
\begin{align}
\label{SolHam-Irr-PowerLaw}
    A_{n,\tilde{n}} = a^n \, b^{\tilde{n}} \, c^{n \tilde{n}}
    \;,
\end{align}
where $a,b,c \in \C$ are three complex numbers that are independent of $n$ and $\tilde{n}$. If we substitute this ansatz in the recurrence relation \eqref{RecursionEqA}, we get
\begin{align}
\label{SolHam-Irr-Rec-PowerLaw}
    E_{(\lambda,\tambda)} &=
    \frac{\hbar^2}{2m \lambda^2} \pr{
        2 - a \, c^{\tilde{n}} - \frac{1}{a \, c^{\tilde{n}}} }
	\nonumber \\ & \hspace{0.5cm}
    + \frac{m \omega^2 \hbar^2}{2 \tambda^2} \pr{ 2
        - \pr{b \, e^{-ix_0\tambda/\hbar}}
            \pr{c \, e^{-i \lambda \tambda / \hbar}}^n
        - \pr{b \, e^{-ix_0\tambda/\hbar}}^{-1}
            \pr{c \, e^{-i \lambda \tambda / \hbar}}^{-n} }
    \;.
\end{align}
The left-hand side of \eqref{SolHam-Irr-Rec-PowerLaw} is independent of $n$ and $\tilde{n}$, thus the right-hand side must also be independent of these variables. This requires
\begin{align}
\label{SolHam-Irr-FirstFind}
    c = 1
    \qquad \text{and} \qquad
    \frac{\lambda \tambda}{2 \pi \hbar} \in \Z
    \;.
\end{align}
Hence, the solutions of the recurrence relation \eqref{RecursionEqA} that follow the power-law ansatz \eqref{SolHam-Irr-PowerLaw} exist only under the conditions \eqref{SolHam-Irr-FirstFind}. This motivates us to consider the condition $\lambda \tambda \in 2\pi\hbar \, \Z$ in general without the power-law ansatz.

Let's multiply both sides of the equation \eqref{RecursionEqA} with a factor of $e^{-i(n\tilde{r} - \tilde{n}r)}$ for arbitrary $r,\tilde{r} \in \R$, and sum the resulting expression over all $n,\tilde{n} \in \Z$. Assuming that the sum
\begin{align}
\label{PolymericFourier-Irr}
    \phi(r,\tilde{r}) \equiv
    \sum_{n,\tilde{n} \in \Z}
    e^{-i(n\tilde{r} - \tilde{n}r)} \,
    A_{n,\tilde{n}}
\end{align}
has a finite value, we obtain
\begin{align}
\label{Fake_PDE}
    E_{(\lambda,\tambda)} \, \phi(r,\tilde{r})
&=	\frac{\hbar^2}{m \lambda^2}
	\pr{1 - \cos\tilde{r}}
	\phi(r,\tilde{r})
	+ \frac{m \omega^2 \hbar^2}{\tambda^2}
	\pr{1 - \cos\!\pr{r + \tambda x_0 / \hbar}}
	\phi(r,\tilde{r})
\;.
\end{align}
One can notice that this is not a differential equation, unlike in the case of standard polymer quantization \cite{Ashtekar:2002sn,G.:2013lia}. The factor $\phi(r,\tilde{r})$ therefore simply drops from this equation and we find
\begin{align}
\label{BeautifulEnergy}
    E_{(\lambda,\tambda)}
&=	\frac{\hbar^2}{m \lambda^2}
	\pr{1 - \cos\tilde{r}}
	+ \frac{m \omega^2 \hbar^2}{\tambda^2}
	\pr{1 - \cos\!\pr{r + \tambda x_0 / \hbar}}
\;.
\end{align}
Hence, the spectrum of the Hamiltonian \eqref{HamRef} is continuous, and it is bounded from both below and above, such that
\begin{align}
    0 \leq E_{(\lambda,\tambda)} \leq
    \frac{2\hbar^2}{m \lambda^2}
    + \frac{2m \omega^2 \hbar^2}{\tambda^2}
    \;.
\end{align}
The upper bound diverges as $\lambda \rightarrow 0$ or $\tambda \rightarrow 0$.

As the state $\phi(r,\tilde{r})$ cancelled out from the equation \eqref{Fake_PDE}, we found the energy spectrum of the system without finding the states that realize these energy eigenvalues. In fact, the solution
\begin{align}
\label{SolHam-Irr-PlaneWave}
    \psi =
    \sum_{n,\tilde{n} \in \Z}
    e^{i (n \tilde{r} - \tilde{n} r)} \,
    \varphi_{\X_0+(n\lambda,\tilde{n}\tambda)}
\end{align}
for arbitrary $r,\tilde{r} \in \R$ would be an eigenstate of the Hamiltonian \eqref{HamRef} with the energy eigenvalue given in \eqref{BeautifulEnergy}. However, the state in \eqref{SolHam-Irr-PlaneWave} is not normalizable, and therefore it is not an element of the MP Hilbert space.

This situation is familiar from Schr\"odinger quantum mechanics, where the plane waves
\begin{align}
    \ket{\Psi_k} =
    \int_{\R} \dd x \; e^{-ikx} \ket{x}
    \;, \quad k \in \R
    \;,
\end{align}
are not in $L^2(\R)$, nevertheless they serve as a useful tool to build normalizable states. One can treat the solutions in \eqref{SolHam-Irr-PlaneWave} analogously to the plane waves. In this analogy, the local MP basis vectors $\varphi_{\X_0 + (n\lambda,\tilde{n}\tambda)}$ match the Schr\"odinger position eigenvectors $\ket{x}$, the discrete labels $n,\tilde{n} \in \Z$ match the continuous position variable $x \in \R$, and the parameters $r,\tilde{r} \in \R$ match the wave vector $k \in \R$.

In order to investigate the plane wave limit of the energy in \eqref{BeautifulEnergy}, consider setting $\omega = 0$, so that the second term in the Hamiltonian \eqref{HamRef} disappears, and taking the limit $\lambda \rightarrow 0$ while the ratio $\tilde{r} / \lambda$ is held constant. In this limit the energy becomes
\begin{align}
    E_{(\lambda,\tambda)}
    \rightarrow
    \frac{\pr{\tilde{r} \hbar / \lambda}^2}{2m}
    \;,
\end{align}
where $\tilde{r} \hbar / \lambda \in \R$ can be interpreted as the momentum and we obtain the well-known formula for kinetic energy.

\subsection{Rational scale ratios}

Next, we consider the case when both $\lambda/\ell$ and $\tambda/\tell$ are rational numbers, and once again under the assumption $\lambda \tambda \in 2\pi\hbar \, \Z$. For the sake of definiteness, let us write
\begin{align}
\label{EBOB}
N \lambda = M \ell
\qquad \text{and} \qquad
\tilde{N} \tambda = \tilde{M} \tell
\;,
\end{align}
for $N,\tilde{N},M,\tilde{M} \in \Z_+$ where the pairs $(N,M)$ and $(\tilde{N},\tilde{M})$ are coprime.

The key difference with the first case treated above is that the lattice $\LS_{\X_0}^{\Lambda,(\lambda,\tambda)} \subset T_\Lambda$ is now finite. Hence, the expansion of an element $\psi \in \mathcal{H}^{\Lambda, (\lambda,\tambda), \X_0}_\mathrm{MP}$ as in \eqref{PsiInA} contains only finitely many terms, and can be written as
\begin{align}
\label{PsiInA2}
\psi
&=	\sum_{n = 0}^{N-1}
	\sum_{\tilde{n} = 0}^{\tilde{N}-1}
	A_{n,\tilde{n}} \,
	\varphi_{\X_0 + (n\lambda,\tilde{n}\tambda)}
\;.
\end{align}
In this expression, the coefficients $A_{n,\tilde{n}}$ are defined a priori only for $0 \leq n \leq N-1$ and $0 \leq \tilde{n} \leq \tilde{N}-1$. This prevents us from freely rearranging the summation variables in order to factor out $\varphi$ when acting with the Hamiltonian operator.

A lengthy rewriting of the action of the Hamiltonian on the state \eqref{PsiInA2} shows however that we can consistently define the coefficients $A_{n,\tilde{n}}$ for all $n,\tilde{n} \in \Z$ by imposing
%%%%%%%%%%%%%%%%%%%%%%%%%%%%%%
\endnote{For this lengthy rewriting, we act with the regularized Hamiltonian operator on the state \eqref{PsiInA2}, then rearrange and relabel the terms of the sum, and finally use the quasi-periodicity of $\varphi$. This gives
\begin{align}
    \hat{H}_{(\lambda,\tambda)} \, \psi &=
    \sum_{n = 0}^{N-1}
	\sum_{\tilde{n} = 0}^{\tilde{N}-1}
	A_{n,\tilde{n}} \, \bigg(
	\frac{\hbar^2}{2m\lambda^2}
	\pr{2 \varphi_{\X_0 + (n\lambda,\tilde{n}\tambda)}
	- \varphi_{\X_0 + ((n+1)\lambda,\tilde{n}\tambda)}
	- \varphi_{\X_0 + ((n-1)\lambda,\tilde{n}\tambda)} }
	\nonumber \\ &\hspace{3cm}
	+ \frac{m \omega^2 \hbar^2}{2 \tambda^2} \,\Big(
	2 \varphi_{\X_0 + (n\lambda,\tilde{n}\tambda)}
	- e^{i \tambda (x_0 + n \lambda) / \hbar} \,
	\varphi_{\X_0 + (n\lambda,(\tilde{n}+1)\tambda)}
	\nonumber \\ &\hspace{6.9cm}
	- e^{-i \tambda (x_0 + n \lambda) / \hbar} \,
	\varphi_{\X_0 + (n\lambda,(\tilde{n}-1)\tambda)}
	\Big) \bigg)
\displaybreak[0]\nonumber \\ &=
	\sum_{n = 0}^{N-1}
	\sum_{\tilde{n} = 0}^{\tilde{N}-1}
	A_{n,\tilde{n}} \pr{ \frac{\hbar^2}{m\lambda^2}
	+ \frac{m \omega^2 \hbar^2}{\tambda^2} }
	\varphi_{\X_0 + (n\lambda,\tilde{n}\tambda)}
	\nonumber \\ &\hspace{0.5cm}
    - \frac{\hbar^2}{2m\lambda^2} \pr{
    \sum_{n = 0}^{N-2}
	\sum_{\tilde{n} = 0}^{\tilde{N}-1}
	A_{n,\tilde{n}} \,
	\varphi_{\X_0 + ((n+1)\lambda,\tilde{n}\tambda)}
    + A_{N-1,\tilde{n}} \,
	\varphi_{\X_0 + (N\lambda,\tilde{n}\tambda)}
	}
	\nonumber \\ &\hspace{0.5cm}
    - \frac{\hbar^2}{2m\lambda^2} \pr{
    \sum_{n = 1}^{N-1}
	\sum_{\tilde{n} = 0}^{\tilde{N}-1}
	A_{n,\tilde{n}} \,
	\varphi_{\X_0 + ((n-1)\lambda,\tilde{n}\tambda)}
    + A_{0,\tilde{n}} \,
	\varphi_{\X_0 + (-\lambda,\tilde{n}\tambda)}
	}
	\nonumber \\ &\hspace{0.5cm}
	- \frac{m \omega^2 \hbar^2}{2\tambda^2} \pr{
	\sum_{n = 0}^{N-1}
	\sum_{\tilde{n} = 0}^{\tilde{N}-2}
	A_{n,\tilde{n}} \,
	e^{i \tambda (x_0 + n \lambda) / \hbar} \,
	\varphi_{\X_0 + (n\lambda,(\tilde{n}+1)\tambda)}
	+ A_{n,\tilde{N}-1} \,
	e^{i \tambda (x_0 + n \lambda) / \hbar} \,
	\varphi_{\X_0 + (n\lambda,\tilde{N}\tambda)}
	}
	\nonumber \\ &\hspace{0.5cm}
	- \frac{m \omega^2 \hbar^2}{2\tambda^2} \pr{
	\sum_{n = 0}^{N-1}
	\sum_{\tilde{n} = 1}^{\tilde{N}-1}
	A_{n,\tilde{n}} \,
	e^{-i \tambda (x_0 + n \lambda) / \hbar} \,
	\varphi_{\X_0 + (n\lambda,(\tilde{n}-1)\tambda)}
	+ A_{n,0} \,
	e^{-i \tambda (x_0 + n \lambda) / \hbar} \,
	\varphi_{\X_0 + (n\lambda,-\tambda)}
	}
\displaybreak[0]\nonumber \\ &=
	\sum_{n = 0}^{N-1}
	\sum_{\tilde{n} = 0}^{\tilde{N}-1}
	A_{n,\tilde{n}} \pr{ \frac{\hbar^2}{m\lambda^2}
	+ \frac{m \omega^2 \hbar^2}{\tambda^2} }
	\varphi_{\X_0 + (n\lambda,\tilde{n}\tambda)}
	\nonumber \\ &\hspace{0.5cm}
    - \frac{\hbar^2}{2m\lambda^2} \pr{
    \sum_{n = 1}^{N-1}
	\sum_{\tilde{n} = 0}^{\tilde{N}-1}
	A_{n-1,\tilde{n}} \,
	\varphi_{\X_0 + (n\lambda,\tilde{n}\tambda)}
    + A_{N-1,\tilde{n}} \,
	\varphi_{\X_0 + (0,\tilde{n}\tambda) + (M\ell,0)}
	}
	\nonumber \\ &\hspace{0.5cm}
    - \frac{\hbar^2}{2m\lambda^2} \pr{
    \sum_{n = 0}^{N-2}
	\sum_{\tilde{n} = 0}^{\tilde{N}-1}
	A_{n+1,\tilde{n}} \,
	\varphi_{\X_0 + (n\lambda,\tilde{n}\tambda)}
    + A_{0,\tilde{n}} \,
	\varphi_{\X_0 + ((N-1)\lambda,\tilde{n}\tambda) + (-M\ell,0)}
	}
	\nonumber \\ &\hspace{0.5cm}
	- \frac{m \omega^2 \hbar^2}{2\tambda^2} \pr{
	\sum_{n = 0}^{N-1}
	\sum_{\tilde{n} = 1}^{\tilde{N}-1}
	A_{n,\tilde{n}-1} \,
	e^{i \tambda (x_0 + n \lambda) / \hbar} \,
	\varphi_{\X_0 + (n\lambda,\tilde{n}\tambda)}
	+ A_{n,\tilde{N}-1} \,
	e^{i \tambda (x_0 + n \lambda) / \hbar} \,
	\varphi_{\X_0 + (n\lambda,0) + (0,\tilde{M}\tell)}
	}
	\nonumber \\ &\hspace{0.5cm}
	- \frac{m \omega^2 \hbar^2}{2\tambda^2} \pr{
	\sum_{n = 0}^{N-1}
	\sum_{\tilde{n} = 0}^{\tilde{N}-2}
	A_{n,\tilde{n}+1} \,
	e^{-i \tambda (x_0 + n \lambda) / \hbar} \,
	\varphi_{\X_0 + (n\lambda,\tilde{n}\tambda)}
	+ A_{n,0} \,
	e^{-i \tambda (x_0 + n \lambda) / \hbar} \,
	\varphi_{\X_0 + (n\lambda,(\tilde{N}-1)\tambda) + (0,-\tilde{M}\tell)}
	}
\displaybreak[0]\nonumber \\ &=
	\sum_{n = 0}^{N-1}
	\sum_{\tilde{n} = 0}^{\tilde{N}-1}
	A_{n,\tilde{n}} \pr{ \frac{\hbar^2}{m\lambda^2}
	+ \frac{m \omega^2 \hbar^2}{\tambda^2} }
	\varphi_{\X_0 + (n\lambda,\tilde{n}\tambda)}
	\nonumber \\ &\hspace{0.5cm}
    - \frac{\hbar^2}{2m\lambda^2} \pr{
    \sum_{n = 1}^{N-1}
	\sum_{\tilde{n} = 0}^{\tilde{N}-1}
	A_{n-1,\tilde{n}} \,
	\varphi_{\X_0 + (n\lambda,\tilde{n}\tambda)}
    + A_{N-1,\tilde{n}} \,
	e^{-i N \lambda (\tilde{x}_0 + \tilde{n}\tambda)/\hbar} \,
	\varphi_{\X_0 + (0,\tilde{n}\tambda)}
	}
	\nonumber \\ &\hspace{0.5cm}
    - \frac{\hbar^2}{2m\lambda^2} \pr{
    \sum_{n = 0}^{N-2}
	\sum_{\tilde{n} = 0}^{\tilde{N}-1}
	A_{n+1,\tilde{n}} \,
	\varphi_{\X_0 + (n\lambda,\tilde{n}\tambda)}
    + A_{0,\tilde{n}} \,
	e^{i N \lambda (\tilde{x}_0 + \tilde{n}\tambda)/\hbar} \,
	\varphi_{\X_0 + ((N-1)\lambda,\tilde{n}\tambda)}
	}
	\nonumber \\ &\hspace{0.5cm}
	- \frac{m \omega^2 \hbar^2}{2\tambda^2} \pr{
	\sum_{n = 0}^{N-1}
	\sum_{\tilde{n} = 1}^{\tilde{N}-1}
	A_{n,\tilde{n}-1} \,
	e^{i \tambda (x_0 + n \lambda) / \hbar} \,
	\varphi_{\X_0 + (n\lambda,\tilde{n}\tambda)}
	+ A_{n,\tilde{N}-1} \,
	e^{i \tambda (x_0 + n \lambda) / \hbar} \,
	\varphi_{\X_0 + (n\lambda,0)}
	}
	\nonumber \\ &\hspace{0.5cm}
	- \frac{m \omega^2 \hbar^2}{2\tambda^2} \pr{
	\sum_{n = 0}^{N-1}
	\sum_{\tilde{n} = 0}^{\tilde{N}-2}
	A_{n,\tilde{n}+1} \,
	e^{-i \tambda (x_0 + n \lambda) / \hbar} \,
	\varphi_{\X_0 + (n\lambda,\tilde{n}\tambda)}
	+ A_{n,0} \,
	e^{-i \tambda (x_0 + n \lambda) / \hbar} \,
	\varphi_{\X_0 + (n\lambda,(\tilde{N}-1)\tambda)}
	}.
\label{App-RatHamAct}
\end{align}
With this rewriting of the action of the Hamiltonian, we can see that the coefficients $A_{n,\tilde{n}}$ can be defined to satisfy
\begin{subequations}
\begin{align}
    A_{-1,\tilde{n}} &\equiv
    A_{N-1,\tilde{n}} \,
	e^{-i N \lambda (\tilde{x}_0 + \tilde{n}\tambda)/\hbar},
	\\
    A_{N,\tilde{n}} &\equiv
    A_{0,\tilde{n}} \,
	e^{i N \lambda (\tilde{x}_0 + \tilde{n}\tambda)/\hbar},
	\\
    A_{n,-1} &\equiv A_{n,\tilde{N}-1},
    \\
    A_{n,\tilde{N}} &\equiv A_{n,0},
\end{align}
which can be written more compactly as the general definition \eqref{DefCoef}. Then, \eqref{App-RatHamAct} leads to \eqref{RatHamAct}.
\end{subequations}
}
%%%%%%%%%%%%%%%%%%%%%%%%%%%%%%
\begin{align}
\label{DefCoef}
    A_{n+N,\tilde{n}} \equiv
    e^{i N \lambda \, \tilde{x}_0 /\hbar} \,
    A_{n,\tilde{n}}
    \qquad \text{and} \qquad
    A_{n,\tilde{n}+\tilde{N}} \equiv
    A_{n,\tilde{n}}
    \;.
\end{align}
This implies in particular that
\begin{align}
A_{n+N,\tilde{n}} \,
	\varphi_{\X_0 + ((n+N)\lambda,\tilde{n}\tambda)}
=	A_{n,\tilde{n}} \,
	\varphi_{\X_0 + (n\lambda,\tilde{n}\tambda)}
\quad\text{and}\quad
A_{n,\tilde{n}+\tilde{N}} \,
	\varphi_{\X_0 + (n\lambda,(\tilde{n}+\tilde{N})\tambda)}
=	A_{n,\tilde{n}} \,
	\varphi_{\X_0 + (n\lambda,\tilde{n}\tambda)}
\;.
\end{align}
With this definition, we find
\begin{align}
\label{RatHamAct}
    \hat{H}_{(\lambda,\tambda)} \, \psi &=
	\sum_{n = 0}^{N-1}
	\sum_{\tilde{n} = 0}^{\tilde{N}-1} \bigg(
		\frac{\hbar^2}{2m\lambda^2}
		\pr{2 A_{n,\tilde{n}}
			- A_{n-1,\tilde{n}} - A_{n+1,\tilde{n}}
			}
	\nonumber \\ & \hspace{2.2cm}
		+ \frac{m \omega^2 \hbar^2}{2\tambda^2}
		\pr{2 A_{n,\tilde{n}}
			- e^{i \tambda (x_0 + n \lambda) / \hbar} \,
			A_{n,\tilde{n}-1}
			- e^{-i \tambda (x_0 + n \lambda) / \hbar} \,
			A_{n,\tilde{n}+1}
			}
		\!\! \bigg)
		\, \varphi_{\X_0 + (n\lambda,\tilde{n}\tambda)}\;.
\end{align}
For the eigenstates $\psi$ of the Hamiltonian $\hat{H}_{(\lambda,\tambda)}$, we obtain the relation
\begin{align}
\label{RecursionEqA-Rat}
E_{(\lambda,\tambda)} \, A_{n,\tilde{n}}
&=	\frac{\hbar^2}{2m\lambda^2}
	\pr{2 A_{n,\tilde{n}}
		- A_{n-1,\tilde{n}} - A_{n+1,\tilde{n}}
		}
	\nonumber \\ & \hspace{0.5cm}
	+ \frac{m \omega^2 \hbar^2}{2\tambda^2}
	\pr{2 A_{n,\tilde{n}}
		- e^{i \tambda x_0 / \hbar} \,
		A_{n,\tilde{n}-1}
		- e^{-i \tambda x_0 / \hbar} \,
		A_{n,\tilde{n}+1}
		}
\;.
\end{align}
 This equation is identical to the recurrence relation \eqref{RecursionEqA} in the previous case (using the assumption \eqref{AssumptionLambda}), but its solutions are also constrained by \eqref{DefCoef}. If we consider the solutions of the form
\begin{align}
\label{Rat-Ansatz}
    A_{n,\tilde{n}} = e^{i (n \tilde{r} - \tilde{n} r)}
\end{align}
for $r,\tilde{r} \in \R$ as before, we find the restrictions
\begin{align}
\label{RestrictR}
    \frac{\tilde{N} r}{2\pi} \in \Z
    \qquad \text{and} \qquad
    \frac{N\tilde{r}}{2\pi}
    - \frac{N\lambda\tilde{x}_0}{2\pi\hbar} \in \Z
    \;.
\end{align}
Since these solutions are invariant under shifting the parameters $r, \tilde{r}$ by a multiple of $2\pi$, we find exactly $N \times \tilde{N}$ different solutions of the form \eqref{Rat-Ansatz}. Hence, we can write
\begin{align}
    r = \frac{2\pi k}{\tilde{N}}
    \qquad \text{and} \qquad
    \tilde{r} =
    \frac{2\pi\tilde{k}}{N}
    + \frac{\lambda\tilde{x}_0}{\hbar}
\end{align}
for $k = 1,...,\tilde{N}$ and $\tilde{k} = 1,...,N$. We obtain the energy eigenstates
\begin{align}
\label{SolHam-Rat-PlaneWave-New}
    \psi =
    \sum_{n = 0}^{N-1}
	\sum_{\tilde{n} = 0}^{\tilde{N}-1}
    e^{2\pi i (\tilde{k} n / N - k \tilde{n} / \tilde{N})} \,
	e^{i n \lambda\tilde{x}_0 / \hbar} \,
    \varphi_{\X_0+(n\lambda,\tilde{n}\tambda)}
\end{align}
with the corresponding energy eigenvalues
\begin{align}
\label{BeautifulEnergy2}
    E_{(\lambda,\tambda)}
&=	\frac{\hbar^2}{m \lambda^2}
	\pr{1 - \cos \bigg[ 2\pi \bigg(
	\frac{\tilde{k}}{N} + \frac{\tilde{x}_0}{\tambda}
	\bigg) \bigg] }
	+ \frac{m \omega^2 \hbar^2}{\tambda^2}
	\pr{1 - \cos \bigg[ 2\pi \bigg(
	\frac{k}{\tilde{N}} + \frac{x_0}{\lambda}
	\bigg)\bigg] }
\;.
\end{align}
The spectrum of the Hamiltonian operator in the superselection sector $\mathcal{H}^{\Lambda, (\lambda,\tambda), \X_0}_\mathrm{MP}$ consists of these $N \times \tilde{N}$ discrete values.
\endnote{Since the shift of the parameter $\X_0$ by a lattice vector maps the solutions we found to the same superselection sector of the Hilbert space, one can question whether we correctly identified all possible solutions under the given assumptions, or whether these solutions should have $M \times \tilde{M}$ different copies. Let's name our solutions as
\begin{align}
\psi_{\X_0,k,\tilde{k}} &\equiv
	\sum_{n = 0}^{N-1}
	\sum_{\tilde{n} = 0}^{\tilde{N}-1}
    e^{2\pi i (\tilde{k} n / N - k \tilde{n} / \tilde{N})} \,
    e^{i n \lambda \tilde{x}_0 / \hbar} \,
	\varphi_{\X_0 + (n\lambda,\tilde{n}\tambda)}
\;.
\end{align}
Precisely, the question we investigate here is whether the sets
\begin{align*}
    \cur{ \psi_{\X_0,k,\tilde{k}} \,\big\vert\,
    \tilde{k} = 1,...,N \;;\; k = 1,...,\tilde{N} }
    \quad \mathrm{and} \quad
    \cur{ \psi_{\X_0+\mathbb{M},k,\tilde{k}} \,\big\vert\,
    \tilde{k} = 1,...,N \;;\; k = 1,...,\tilde{N} }
\end{align*}
are linearly dependent for arbitrary values of $\mathbb{M} \in \Lambda$? The answer is positive as one finds
\begin{align}
    \psi_{\X_0 + (m\ell , \tilde{m}\tell) , k , \tilde{k}} &=
    e^{-i m \ell \tilde{x}_0 / \hbar} \,
    \psi_{\X_0 , k + m \tilde{M} , \tilde{k} + \tilde{m} M}
    \;,
\end{align}
for any $m,\tilde{m} \in \Z$. Hence, shifting the anchoring point on the modular space by a lattice vector can be interpreted as a permutation of the solutions up to a phase. This confirms that the $N \times \tilde{N}$ solutions we found are complete.}

If one chooses the coarse-graining scale for the Hamiltonian to be equal to the MP scale, i.e.~$\lambda = \ell$ and $\tambda = \tell$, then there is a single energy eigenstate in each superselection sector that is given by $\psi = e^{i n \lambda \tilde{x}_0 / \hbar} \, \varphi_{\X_0}^{\phantom{O}}$ with the energy
\begin{align}
    E_{(\ell,\tell)} =
    \frac{2\hbar^2}{m \ell^2} \, \sin^2 \!\bigg(
        \frac{\pi \tilde{x}_0}{\tell} \bigg)
    + \frac{2m \omega^2 \hbar^2}{\tell^2} \, \sin^2 \!\bigg( \frac{\pi x_0}{\ell} \bigg)
    \;.
\end{align}

There are two important differences between our results here (when $\lambda/\ell$ and $\tambda/\tell$ are rational numbers) and those in the previous subsection (when the ratios are irrational numbers). Firstly, we find here only a finite number of distinct elements in the energy spectrum. Secondly, the plane-wave-like solutions \eqref{SolHam-Rat-PlaneWave-New} are normalizable, unlike \eqref{SolHam-Irr-PlaneWave} in the previous case. The existence of normalizable plane waves is a new feature of the MP representation, which has no analog in the Schr\"odinger or polymer representations.

\subsection{Discussion}

Now that we have applied the MP representation to study the harmonic oscillator, a few comments and remarks are in order. First, one should recall that we have studied the case $\lambda \tambda \in 2\pi\hbar \, \Z$, which was motivated by our power-law ansatz \eqref{SolHam-Irr-PowerLaw}. Within this assumption, we have found that the spectrum of the harmonic oscillator is bounded from below and above, and can be either continuous or discrete depending on the nature of the scale ratios. This can be contrasted with the study of the harmonic oscillator in the polymer representation \cite{Ashtekar:2002sn,G.:2013lia}. In position polarization, the $p$-polymer representation leads to a difference equation for the action of the regularized Hamiltonian on a state. This can be mapped in the momentum polarization to a differential equation (in the position variable) known as the Mathieu equation, whose solutions can be analyzed to find a discrete spectrum. As pointed out in \cite{G.:2013lia}, studying the spectrum of the harmonic oscillator in the polymer representation can be mapped to the study of a periodic potential in the Schr\"odinger representation.

An important difference between the polymer and MP representations is the following. In the polymer representation, there exists a Fourier transform mapping the (say) $p$-polymer representation in position polarization to that in momentum representation. This is precisely the transformation turning the difference equation into the Mathieu equation. In the MP representation however, since neither $q$ nor $p$ exist as operators, and since there is only a single polarization, namely the modular one, it is not possible to turn the difference equation obtained by acting with the Hamiltonian on a state into a differential equation. As we have seen however, it is possible to write a transformation \eqref{PolymericFourier-Irr} which enables to find the energy eigenvalues without solving for the eigenstate. This situation is similar to what happens for a Schr\"odinger free particle, where in the analogue of \eqref{Fake_PDE} the state also drops out. In future work we will push the analogy between the MP representation and the polymer one further, in particular by defining the MP equivalent of the space Cyl of cylindrical functions and its algebraic dual Cyl$^\star$, together with the action of the latter on the former. While we have not needed this structure for our study, it could prove important for studying the case $\lambda \tambda \notin 2\pi \hbar \, \Z$ and systems beyond the simple harmonic oscillator.

\section{Perspectives}
\label{sec:6}

In this paper we have constructed a new family of representations of the Weyl algebra of exponentiated position and momentum operators. To achieve this, we have first recalled in section \ref{sec:2} how the Schr\"odinger, polymer, and modular representations are constructed, and then proceeded to a polymerization of the modular representation in section \ref{sec:3}. Our new family of representations therefore relies on the same ingredients as the polymer and modular representations, namely the existence of a length scale used to unify position and momentum in the form of the modular lattices, and then the so-called irregularity condition, which consists in relaxing the weak continuity of the Weyl operators. By doing so, we have obtained a family of modular polymer representations labelled by a choice of modular lattice, and then proved the equivalence of representations within the modular polymer family in section \ref{sec:4}. Finally in section \ref{sec:5} we have applied this modular polymer representation to the study of the harmonic oscillator. There, we have seen how the energy spectrum depends on whether the ratios $\lambda/\ell$ and $\tambda/\tell$ are rational or irrational, where $\lambda$ is the scale used to construct the exponentiated Weyl operators, $\ell$ the scale used to define the modular lattice, and the tilded scales are their duals. More precisely, we have studied the case in which these rations are either both irrational or rational. In the first case the spectrum is continuous and bounded, while in the second case it is discrete. This is reminiscent of studies of non-separable Hilbert spaces, where the spectra of e.g. the translation (or exponentiated momentum) operator can be continuous or discrete depending on whether the translation parameter is rational or irrational \cite{Dittrich:2015vfa,Bahr:2015bra}.

As explained in the introduction, our motivation for constructing and studying these modular polymer representations in a finite-dimensional setting was to investigate how inequivalent alternatives to the Schr\"odinger representation can be constructed based on physical inputs from possible theories of quantum gravity and quantum spacetime. On the one hand, LQG makes a strong case for the study of polymer representations, as demonstrated very simply and efficiently by LQC \cite{Ashtekar:2003hd,Brunnemann:2007du,Ashtekar_2011}, while on the other hand work such as \cite{Freidel:2016aa} has suggested that the modular representation is more fundamental than the Schr\"odinger one.

With this in mind, and with the modular polymer representation now at our disposal, the next important step will be to consider possible quantum gravitational models (i.e. different from the simple harmonic oscillator) in which the modular polymer representations could serve to describe new physics. Just like the polymer representation is enough to drastically modify the structure of Wheeler--de Witt quantum cosmology (i.e. quantum cosmology based on the Schr\"odinger representation) and turn it into LQC with the avoidance of the Big Bang singularity, one can naturally expect new physical consequences from the application of the modular polymer representation to quantum cosmology.

Another reason for studying the modular polymer representation is actually grounded in LQG itself. Indeed, the spin network representation mentioned in the introduction, together with its Hilbert space structure which has then inspired LQC, is based on the so-called Ashtekar--Lewandowski representation of LQG \cite{Thiemann:2001yy}. This representation comes with a type of Stone--von Neumann uniqueness theorem, known as the LOST theorem \cite{Lewandowski_2006,Fleischhack_2008} based on the representation of holonomy and flux operators. However, just like the Stone--von Neumann uniqueness theorem of quantum mechanics for the Schr\"odinger representation can be bypassed to construct the polymer representations (as we have recalled in section \ref{sec:2}), the LOST theorem can also be bypassed to construct an inequivalent representation of LQG \cite{Bahr:2015bra,Dittrich:2014wpa,Dittrich:2014wda}. This is known as the BF representation, and its essential distinguishing feature is that it only allows for exponentiated flux operators, and not for the existence of the fluxes themselves as fundamental operators (exponentiated fluxes were also studied in the classical setup in \cite{Rovelli:2015fwa}). Studying the dynamics of LQG in this representation is still an ongoing and involved task, and it would therefore be interesting to try to apply this representation to the study of symmetry-reduced models. Like in LQC, where one takes inspiration from the Ashtekar--Lewandowski representation of LQG to build a quantum theory inequivalent to Wheeler-de Witt quantum cosmology, one could take the ingredients of the BF representation and transpose them to the quantum cosmological setting. This would require to work with exponentiated flux operators, and also with the so-called gauge-covariant fluxes. Interestingly, the use of the gauge-covariant fluxes in the usual representation of LQG has already been investigated in \cite{Liegener:2019zgw,Liegener:2019ymd,Liegener:2019dzj}. Finally, studying the modular polymer representation will also be the occasion of applying and investigating the role of the Zak transform in quantum cosmology. We plan to come back to this in future work.

\section*{Acknowledgements}

We would like to thank Sylvain Carrozza, Bianca Dittrich, Laurent Freidel, Francesco Sartini, and Lee Smolin for discussions and comments. Research at Perimeter Institute is supported in part by the Government of Canada through the Department of Innovation, Science and Economic Development Canada, and by the Province of Ontario through the Ministry of Economic Development, Job Creation and Trade. This research was also partly supported by grants from John Templeton Foundation and FQXi.

\appendix

\section{Notes and proofs}

In this appendix we gather the various notes and proofs from the main text. Clicking on the link in the left margin brings back to the corresponding location in the main text.

\vspace{-1.5cm}

\printendnotes[custom]

\newpage

\bibliography{Biblio.bib}

\providecommand{\href}[2]{#2}\begingroup\raggedright\begin{thebibliography}{10}

\bibitem{Thiemann:2001yy}
T.~Thiemann, \emph{{Modern canonical quantum general relativity}}.
\newblock Univ. Pr., Cambridge, UK, 2007.

\bibitem{Ashtekar:2004eh}
A.~Ashtekar and J.~Lewandowski, \emph{{Background independent quantum gravity:
  A status report}},
  \href{http://dx.doi.org/10.1088/0264-9381/21/15/R01}{\emph{Class. Quant.
  Grav.} {\bfseries 21} (2004) R53},
  [\href{https://arxiv.org/abs/gr-qc/0404018}{{\ttfamily gr-qc/0404018}}].

\bibitem{Oriti:2018dsg}
D.~Oriti, \emph{{Levels of spacetime emergence in quantum gravity}},
  [\href{https://arxiv.org/abs/1807.04875}{{\ttfamily 1807.04875}}].

\bibitem{Penrose_2014}
R.~Penrose, \emph{On the gravitization of quantum mechanics 1: Quantum state
  reduction},
  \href{http://dx.doi.org/10.1007/s10701-013-9770-0}{\emph{Foundations of
  Physics} {\bfseries 44} (Jan, 2014) 557--575}.

\bibitem{Freidel_2014}
L.~Freidel, R.~G. Leigh and D.~Minic, \emph{Quantum gravity, dynamical
  phase-space and string theory},
  \href{http://dx.doi.org/10.1142/s0218271814420061}{\emph{International
  Journal of Modern Physics D} {\bfseries 23} (Oct, 2014) 1442006}.

\bibitem{Howl_2019}
R.~Howl, R.~Penrose and I.~Fuentes, \emph{Exploring the unification of quantum
  theory and general relativity with a bose--einstein condensate},
  \href{http://dx.doi.org/10.1088/1367-2630/ab104a}{\emph{New Journal of
  Physics} {\bfseries 21} (Apr, 2019) 043047}.

\bibitem{Lewandowski_2006}
J.~Lewandowski, A.~Oko{\l}{\'o}w, H.~Sahlmann and T.~Thiemann, \emph{Uniqueness
  of diffeomorphism invariant states on holonomy--flux algebras},
  \href{http://dx.doi.org/10.1007/s00220-006-0100-7}{\emph{Communications in
  Mathematical Physics} {\bfseries 267} (Aug, 2006) 703--733}.

\bibitem{Fleischhack_2008}
C.~Fleischhack, \emph{Representations of the weyl algebra in quantum geometry},
  \href{http://dx.doi.org/10.1007/s00220-008-0593-3}{\emph{Communications in
  Mathematical Physics} {\bfseries 285} (Oct, 2008) 67--140}.

\bibitem{HALVORSON200445}
H.~Halvorson, \emph{Complementarity of representations in quantum mechanics},
  \href{http://dx.doi.org/https://doi.org/10.1016/j.shpsb.2003.01.001}{\emph{Studies
  in History and Philosophy of Science Part B: Studies in History and
  Philosophy of Modern Physics} {\bfseries 35} (2004) 45 -- 56}.

\bibitem{Strocchi_2016}
F.~Strocchi, \emph{{Gauge Invariance and Weyl-polymer Quantization}},
  \href{http://dx.doi.org/10.1007/978-3-319-17695-6}{\emph{Lecture Notes in
  Physics} (2016) }.

\bibitem{Ashtekar:2002sn}
A.~Ashtekar, S.~Fairhurst and J.~L. Willis, \emph{{Quantum gravity, shadow
  states, and quantum mechanics}},
  \href{http://dx.doi.org/10.1088/0264-9381/20/6/302}{\emph{Class.Quant.Grav.}
  {\bfseries 20} (2003) 1031--1062},
  [\href{https://arxiv.org/abs/gr-qc/0207106}{{\ttfamily gr-qc/0207106}}].

\bibitem{Corichi:2007aa}
A.~Corichi, T.~Vukasinac and J.~A. Zapata, \emph{Polymer quantum mechanics and
  its continuum limit},
  \href{http://dx.doi.org/10.1103/PhysRevD.76.044016}{\emph{Physical Review D}
  {\bfseries 76} (2007) }.

\bibitem{Ashtekar_2011}
A.~Ashtekar and P.~Singh, \emph{Loop quantum cosmology: a status report},
  \href{http://dx.doi.org/10.1088/0264-9381/28/21/213001}{\emph{Classical and
  Quantum Gravity} {\bfseries 28} (Sep, 2011) 213001}.

\bibitem{10025994495}
Y.~Kuramoto, \emph{International symposium on mathematical problems in
  theoretical physics}, {\emph{Lecture Notes in Physics} {\bfseries 30} (1975)
  420}.

\bibitem{Acerbi_1993}
F.~Acerbi, G.~Morchio and F.~Strocchi, \emph{Infrared singular fields and
  nonregular representations of canonical commutation relation algebras},
  \href{http://dx.doi.org/10.1063/1.530200}{\emph{Journal of Mathematical
  Physics} {\bfseries 34} (Mar, 1993) 899--914}.

\bibitem{L_ffelholz_1995}
J.~Loffelholz, G.~Morchio and F.~Strocchi, \emph{Spectral stochastic processes
  arising in quantum mechanical models with a non-l 2 ground state},
  \href{http://dx.doi.org/10.1007/bf00761297}{\emph{Letters in Mathematical
  Physics} {\bfseries 35} (Nov, 1995) 251--262}.

\bibitem{NILL_1992}
F.~Nill, \emph{A constructive quantum field theoretic approach to chern-simons
  theory},
  \href{http://dx.doi.org/10.1142/s0217979292001080}{\emph{International
  Journal of Modern Physics B} {\bfseries 06} (Jun, 1992) 2159--2198}.

\bibitem{Hossain:2010eb}
G.~M. Hossain, V.~Husain and S.~S. Seahra, \emph{{The Propagator in polymer
  quantum field theory}},
  \href{http://dx.doi.org/10.1103/PhysRevD.82.124032}{\emph{Phys. Rev.}
  {\bfseries D82} (2010) 124032},
  [\href{https://arxiv.org/abs/1007.5500}{{\ttfamily 1007.5500}}].

\bibitem{Freidel_2015}
L.~Freidel, R.~G. Leigh and D.~Minic, \emph{Modular spacetime},
  \href{http://dx.doi.org/10.1142/s0218271815440289}{\emph{International
  Journal of Modern Physics D} {\bfseries 24} (Oct, 2015) 1544028}.

\bibitem{Freidel:2016aa}
L.~Freidel, R.~G. Leigh and D.~Minic, \emph{Quantum spaces are modular},
  \href{http://dx.doi.org/10.1103/PhysRevD.94.104052}{\emph{Physical Review D}
  {\bfseries 94} (2016) }.

\bibitem{Freidel_2017}
L.~Freidel, R.~G. Leigh and D.~Minic, \emph{Modular spacetime and metastring
  theory},
  \href{http://dx.doi.org/10.1088/1742-6596/804/1/012032}{\emph{Journal of
  Physics: Conference Series} {\bfseries 804} (Jan, 2017) 012032}.

\bibitem{Freidel_2019}
L.~Freidel, R.~G. Leigh and D.~Minic, \emph{Manifest non-locality in quantum
  mechanics, quantum field theory and quantum gravity},
  \href{http://dx.doi.org/10.1142/s0217751x19410045}{\emph{International
  Journal of Modern Physics A} {\bfseries 34} (Oct, 2019) 1941004}.

\bibitem{Aharonov_1969}
Y.~Aharonov, H.~Pendleton and A.~Petersen, \emph{Modular variables in quantum
  theory}, \href{http://dx.doi.org/10.1007/bf00670008}{\emph{International
  Journal of Theoretical Physics} {\bfseries 2} (Oct, 1969) 213--230}.

\bibitem{2003qpqt.book.....A}
Y.~{Aharonov} and D.~{Rohrlich}, \emph{{Quantum Paradoxes: Quantum Theory for
  the Perplexed}}.
\newblock John Wiley \& Sons, 2003.

\bibitem{Zak:1967aa}
J.~Zak, \emph{Finite translations in solid-state physics},
  \href{http://dx.doi.org/10.1103/PhysRevLett.19.1385}{\emph{Physical Review
  Letters} {\bfseries 19} (1967) 1385--1387}.

\bibitem{Zak:1993aa}
J.~Zak, \emph{Finite translations in time and energy},
  \href{http://dx.doi.org/10.1103/PhysRevLett.71.2623}{\emph{Physical Review
  Letters} {\bfseries 71} (1993) 2623--2625}.

\bibitem{Yargic:2020kkb}
Y.~Yargic, \emph{{Path Integral in Modular Space}},
  [\href{https://arxiv.org/abs/2002.01604}{{\ttfamily 2002.01604}}].

\bibitem{FREIDEL2014302}
L.~Freidel, R.~G. Leigh and D.~Minic, \emph{Born reciprocity in string theory
  and the nature of spacetime},
  \href{http://dx.doi.org/https://doi.org/10.1016/j.physletb.2014.01.067}{\emph{Physics
  Letters B} {\bfseries 730} (2014) 302 -- 306}.

\bibitem{Svoboda:2019fpt}
F.~J. Rudolph and D.~Svoboda, \emph{{Born Geometry in a Nutshell}},
  \href{http://dx.doi.org/10.22323/1.347.0126}{\emph{PoS} {\bfseries CORFU2018}
  (2019) 126}, [\href{https://arxiv.org/abs/1904.06989}{{\ttfamily
  1904.06989}}].

\bibitem{Corichi:2006qf}
A.~Corichi, T.~Vukasinac and J.~A. Zapata, \emph{{Hamiltonian and physical
  Hilbert space in polymer quantum mechanics}},
  \href{http://dx.doi.org/10.1088/0264-9381/24/6/008}{\emph{Class. Quant.
  Grav.} {\bfseries 24} (2007) 1495--1512},
  [\href{https://arxiv.org/abs/gr-qc/0610072}{{\ttfamily gr-qc/0610072}}].

\bibitem{G.:2013lia}
J.~F. Barbero~G., J.~Prieto and E.~J.~S. Villaseor, \emph{{Band structure in
  the polymer quantization of the harmonic oscillator}},
  \href{http://dx.doi.org/10.1088/0264-9381/30/16/165011}{\emph{Class. Quant.
  Grav.} {\bfseries 30} (2013) 165011},
  [\href{https://arxiv.org/abs/1305.5406}{{\ttfamily 1305.5406}}].

\bibitem{Stone_1930}
M.~H. Stone, \emph{{Linear Transformations in Hilbert Space: III. Operational
  Methods and Group Theory}},
  \href{http://dx.doi.org/10.1073/pnas.16.2.172}{\emph{Proceedings of the
  National Academy of Sciences} {\bfseries 16} (Feb, 1930) 172--175}.

\bibitem{v_Neumann_1931}
J.~v.~Neumann, \emph{{Die Eindeutigkeit der Schr{\"o}dingerschen Operatoren}},
  \href{http://dx.doi.org/10.1007/bf01457956}{\emph{Mathematische Annalen}
  {\bfseries 104} (Dec, 1931) 570--578}.

\bibitem{Rosenberg}
J.~Rosenberg, \emph{{A Selective History of the Stone-von Neumann Theorem}},
  \href{http://dx.doi.org/10.1090/conm/365/06710}{\emph{{Contemp. Math.}}
  {\bfseries 365} (01, 2004) }.

\bibitem{Summers_2001}
S.~J. Summers, \emph{{On the Stone-von Neumann Uniqueness Theorem and Its
  Ramifications}},
  \href{http://dx.doi.org/10.1007/978-94-017-2012-0_9}{\emph{John von Neumann
  and the Foundations of Quantum Physics} (2001) 135--152}.

\bibitem{Cavallaro1999}
S.~Cavallaro, G.~Morchio and F.~Strocchi, \emph{{A Generalization of the
  Stone--Von Neumann Theorem to Nonregular Representations of the
  CCR-Algebra}}, \href{http://dx.doi.org/10.1023/A:1007599222651}{\emph{Letters
  in Mathematical Physics} {\bfseries 47} (Mar, 1999) 307--320}.

\bibitem{Freidel:2017yuv}
L.~Freidel, F.~J. Rudolph and D.~Svoboda, \emph{{Generalised Kinematics for
  Double Field Theory}},
  \href{http://dx.doi.org/10.1007/JHEP11(2017)175}{\emph{JHEP} {\bfseries 11}
  (2017) 175}, [\href{https://arxiv.org/abs/1706.07089}{{\ttfamily
  1706.07089}}].

\bibitem{Dittrich:2015vfa}
B.~Dittrich, P.~A. Hoehn, T.~A. Koslowski and M.~I. Nelson, \emph{{Chaos, Dirac
  observables and constraint quantization}},
  [\href{https://arxiv.org/abs/1508.01947}{{\ttfamily 1508.01947}}].

\bibitem{Bahr:2015bra}
B.~Bahr, B.~Dittrich and M.~Geiller, \emph{{A new realization of quantum
  geometry}},  [\href{https://arxiv.org/abs/1506.08571}{{\ttfamily
  1506.08571}}].

\bibitem{Ashtekar:2003hd}
A.~Ashtekar, M.~Bojowald and J.~Lewandowski, \emph{{Mathematical structure of
  loop quantum cosmology}}, {\emph{Adv. Theor. Math. Phys.} {\bfseries 7}
  (2003) 233--268}, [\href{https://arxiv.org/abs/gr-qc/0304074}{{\ttfamily
  gr-qc/0304074}}].

\bibitem{Brunnemann:2007du}
J.~Brunnemann and C.~Fleischhack, \emph{{On the Configuration Spaces of
  Homogeneous Loop Quantum Cosmology and Loop Quantum Gravity}},
  [\href{https://arxiv.org/abs/0709.1621}{{\ttfamily 0709.1621}}].

\bibitem{Dittrich:2014wpa}
B.~Dittrich and M.~Geiller, \emph{{A new vacuum for Loop Quantum Gravity}},
  \href{http://dx.doi.org/10.1088/0264-9381/32/11/112001}{\emph{Class. Quant.
  Grav.} {\bfseries 32} (2015) 112001},
  [\href{https://arxiv.org/abs/1401.6441}{{\ttfamily 1401.6441}}].

\bibitem{Dittrich:2014wda}
B.~Dittrich and M.~Geiller, \emph{{Flux formulation of loop quantum gravity:
  Classical framework}},
  \href{http://dx.doi.org/10.1088/0264-9381/32/13/135016}{\emph{Class. Quant.
  Grav.} {\bfseries 32} (2015) 135016},
  [\href{https://arxiv.org/abs/1412.3752}{{\ttfamily 1412.3752}}].

\bibitem{Rovelli:2015fwa}
C.~Rovelli and F.~Vidotto, \emph{{Compact phase space, cosmological constant,
  and discrete time}},
  \href{http://dx.doi.org/10.1103/PhysRevD.91.084037}{\emph{Phys. Rev.}
  {\bfseries D91} (2015) 084037},
  [\href{https://arxiv.org/abs/1502.00278}{{\ttfamily 1502.00278}}].

\bibitem{Liegener:2019zgw}
K.~Liegener and P.~Singh, \emph{{Gauge invariant bounce from quantum
  geometry}},  [\href{https://arxiv.org/abs/1906.02759}{{\ttfamily
  1906.02759}}].

\bibitem{Liegener:2019ymd}
K.~Liegener and P.~Singh, \emph{{New Loop Quantum Cosmology Modifications from
  Gauge-covariant Fluxes}},
  \href{http://dx.doi.org/10.1103/PhysRevD.100.124048}{\emph{Phys. Rev.}
  {\bfseries D100} (2019) 124048},
  [\href{https://arxiv.org/abs/1908.07001}{{\ttfamily 1908.07001}}].

\bibitem{Liegener:2019dzj}
K.~Liegener and P.~Singh, \emph{{Some physical implications of regularization
  ambiguities in SU(2) gauge-invariant loop quantum cosmology}},
  \href{http://dx.doi.org/10.1103/PhysRevD.100.124049}{\emph{Phys. Rev.}
  {\bfseries D100} (2019) 124049},
  [\href{https://arxiv.org/abs/1908.07543}{{\ttfamily 1908.07543}}].

\end{thebibliography}\endgroup
\bibliographystyle{Biblio}

\end{document}